\documentclass[10pt,journal,compsoc,romanappendices]{IEEEtran}
\ifCLASSOPTIONcompsoc
  \usepackage[nocompress]{cite}
\else
  \usepackage{cite}
\fi
\usepackage{tablefootnote}

\usepackage{lineno,hyperref}
\modulolinenumbers[5]

\usepackage{graphicx}
\usepackage{url}
\usepackage{bbm}
\usepackage{amsfonts}
\usepackage{amssymb}
\usepackage{kantlipsum}
\usepackage{enumerate}
\usepackage{amsthm}
\usepackage[T1]{fontenc}
\usepackage{multirow}
\usepackage{amsmath}
\usepackage{bbm}
\usepackage{dsfont}
\usepackage{xcolor}
\usepackage{booktabs}
\usepackage{mathtools}
\usepackage[utf8]{inputenc}
\usepackage{enumerate}
\usepackage{algorithm}
\usepackage{algpseudocode}
\newcommand{\remove}[1]{}
\newtheorem{theorem}{Theorem}[section]
\newtheorem{lemma}{Lemma}[section]
\newtheorem{remark}{Remark}[section]

\algdef{SE}[DOWHILE]{Do}{doWhile}{\algorithmicdo}[1]{\algorithmicwhile\ #1}
\def\e{\epsilon}

\def\E{{\mathbb{E}}}
\def\P{{\mathbb{P}}}
\DeclareMathOperator*{\argmax}{arg\,max}

\algnewcommand{\algorithmicgoto}{\textbf{go to}}%
\algnewcommand{\Goto}[1]{\algorithmicgoto~\ref{#1}}
\begin{document}
\title{Scheduling Policies for Stability and Optimal Server Running Cost in Cloud Computing Platforms}
	
	\author{
		\IEEEauthorblockN{Haritha~K, Chandramani~Singh}
		\IEEEauthorblockA{
			Department of Electronic Systems Engineering\\
			Indian Institute of Science \\
			Bangalore 560012, India \\
			Email: \{haritha, chandra\}iisc.ac.in}
		\thanks{The first author is supported by Visvesvaraya PhD Scheme and the second author by INSPIRE Faculty Research Grant~(DSTO-1363).}
	}

	\maketitle
	\begin{abstract}		
		
		We propose throughput and cost optimal job scheduling algorithms in cloud computing platforms offering Infrastructure as a Service. We first consider online migration and propose job scheduling algorithms to minimize job migration and server running costs. We consider algorithms that assume knowledge of job-size on arrival of jobs. We characterize the optimal cost subject to system stability. We develop a drift-plus-penalty framework based algorithm that can achieve optimal cost arbitrarily closely. Specifically, this algorithm yields a tradeoff between delay and costs. We then relax the job-size knowledge assumption and give an algorithm that uses readily offered service to the jobs. We show that this algorithm gives order-wise identical cost as the job size based algorithm. Later, we consider offline job migration that incurs migration delays. We again present throughput optimal algorithms that minimize server running cost. We illustrate the performance of the proposed algorithms and compare these to the existing algorithms via simulation.
	\end{abstract}
		\section{Introduction}
	Cloud computing has emerged as a successful paradigm for providing computing services to consumers, both enterprises and individuals. Cloud computing uses virtualization technology
	for efficient usage of computation and communication resources. Specifically, it dynamically creates virtual machines (VMs) from physical resources to suit to consumers' needs and
	allocates VMs on demand. Availability of low cost computers, servers, storage devices, and high capacity networks have further fuelled the growth of cloud computing clusters. According
	to Gartner Inc.~\cite{Gartner}, the worldwide public cloud services market is projected to grow $ 21.4 $\% to total \$$186.4$~billion by the year $ 2023 $.
	
	There has been substantial work on job scheduling algorithms to enhance cloud computing clusters' throughput and Quality of Service~(QoS).
	However, massive growth in cloud services has also led to a significant increase in  the energy consumption of the clusters. According to a United States data center energy usage report, this energy consumption is expected to rise annually by $ 4 $\% and would be around $ 90 $ billion kWh in $ 2025 $ in US data centers alone~\cite{US-DC-Energy}. Cloud computing data centers have also become major contributors to the global carbon emission footprints. The sources of energy consumption in cloud computing clusters include servers, cooling, communication, storage, and power distribution equipment~(PDE). So, the focus now has turned to algorithms to improve the energy efficiency of each of these components.
	
	We focus on minimization of server running and job migration costs in cloud computing centers. We model the power consumption at the server as the affine function of power consumed by the different types of jobs scheduled on the server. Energy efficient scheduling algorithms often switch off underutilized servers and stall the jobs in service in order to save on the server running
	cost. However, preempting jobs necessitates storing their states and again loading them when their services are resumed~(at potentially different servers).
	We refer to this process of stalling and resuming job services as job migration and the additional resources for storing and loading the states as job migration cost.
	We propose job scheduling algorithms that are throughput optimal and also minimize average server running and job migration costs in cloud computing centers.
	
	\subsection{Related Work}
	Resource allocation algorithms maintain QoS while economizing on operating costs have garnered significant attention in the context of data centers. Maguluri et al.~\cite{stheja2012loadbal,stheja-unknown} proposed {\it max-weight} based throughput optimal job scheduling algorithms in cloud computing clusters. They proposed both preemptive and non-preemptive scheduling algorithms. In~\cite{stheja2012loadbal}, they assumed availability of job size information whereas in~\cite{stheja-unknown}, they considered jobs with unknown sizes. Ghaderi~\cite{Ghaderi-Random} proposed randomized non-preemptive job scheduling algorithms to alleviate the complexity of max-weight based algorithms. To establish the throughput optimality of these algorithms, he relied on connection with loss systems. These algorithms do not require synchronization among servers but are limited to Poisson job arrivals and exponential service durations. Ghaderi et al.~\cite{Ghaderi2015Storms} studied scheduling of jobs, represented as graphs, over a large cluster of servers, again using connection with loss systems. Here also the authors proposed a non-preemptive randomized throughput optimal algorithm. Psychas and Ghaderi~\cite{Psychas:2017:NVS:3175501.3154493},~\cite{2018-Psychas} also proposed a parameterized class of non-preemptive scheduling algorithms that can support certain fractions of the capacity region. Further, the parameter can be tuned to provide a tradeoff between achievable throughput, delay, and computational complexity. The non-preemptive algorithms in these works do not incur any job migration cost, but also do not account for server running cost. 
	\par Psychas and Ghaderi~\cite{2021-Psychas} considered a more practical setting where job's resource requirements belong to a very large set of diverse types, or in the extreme case even infinitely many types. Note that, the classical max-weight algorithms assume finite set of job types and prove throughput optimality. In this work, the authors characterize a fundamental limit on the maximum throughput under the setting considered and develop non-preemptive scheduling algorithms, based on Best-Fit and Universal Partitioning. These algorithms achieve at least $ 1/2 $ and $ 2/3 $ of the maximum throughput respectively.	
	
	Stolyar and Zhong~\cite{stoyler2013binpack} considered VM placement in infinite server systems. They aimed at minimizing the total number of occupied servers at any time but disregarded VM migration cost. They proposed a greedy randomized algorithm that is asymptotically optimal as the job arrival rate approaches infinity. Shi et al.~\cite{Lei_} studied VM placement in finite server systems, modeling it as a vector bin packing problem. They first considered a finite set of jobs and gave an offline algorithm to minimize the number of occupied servers. Then they considered a dynamic setup and proposed an online algorithm to minimize job migration cost as well in addition to the server cost. In their dynamic setup, jobs that do not get required resource on arrival are not queued but leave the system. Wang et al.~\cite{Wang.M-stoch-bin-pack} considered VMs with normal distributed bandwidth requirements. They also aimed at minimizing the number of required servers and formulated it as a stochastic bin packing problem. They proposed an online algorithm yielding number of servers within $(1+\sqrt{2})$ of the optimum.
	
	Our proposed algorithms allow a tradeoff between delay, server running cost and job migration cost through tunable weighing parameters.
	
	Jiang et al.~\cite{JIang2012JointPlaceent} studied a joint VM placement and routing problem aiming to jointly optimize node~(e.g., switch) and link utilization costs. They used Markovian approximation to give an efficient online algorithm. Wang et al.~\cite{Wang2014GreenDCN} introduced a framework to jointly determine VM placement and traffic routing, aiming at minimization of power consumption of all the switches in the network. They argued that the joint problem is NP-hard and proposed a few principles to arrive at an efficient solution.
	
	Feng et al.~\cite{Feng2016DCNC} considered the problem of service distribution in distributed cloud networks. Here each service was mapped to a chain of VMs. They considered joint scheduling of processing and transmission, aiming at minimizing setup and usage related costs of the nodes and the links while stabilizing the network. To this end, they used {\it Lyapunov drift plus penalty} based approach. In~\cite{7511591}, they enhanced the earlier algorithm with a {\it shortest transmission-plus-processing distance} bias that improved delay without compromising throughput and cost.
	
	\par Wu et al.\cite{2019-Wu} and Li Jirui et al.~\cite{2021-Li-Jirui} studied VM placement problem aiming to optimize migration cost, power consumption at servers, and QoS in heterogeneous cloud clusters. The authors employed improved grouping genetic algorithm (IGGA)~\cite{2019-Wu}, improved genetic algorithms~\cite{2021-Li-Jirui} in designing the cost optimal algorithms.

	\par Wang et al.~\cite{Wang2015Optical} considered dynamic scheduling for optical switches in a flow based model of data centers. They took into account reconfiguration delays of the switches and derived a class of adaptive max-weight based algorithms that are provably throughput optimal. Wang et al.\cite{chang2018DCNC} extended this study to scheduling for switches in a network of geographically distributed data centers. They considered server running cost, switch reconfiguration cost and also reconfiguration delays. They also gave a class of throughput optimal algorithms that can be tuned to achieve a tradeoff between delay and the costs. Notice that our work is different from these in that we consider a non-persistent, job based model.
	
	Krishnasamy et al.~\cite{subhashini2017augment} studied joint base station~(BS) activation and rate allocation in cellular networks. They focused on algorithms that optimize (a)~BS activation cost and (b)~BS switching cost, subject to network stability. Our usage of Markov-static policies is closely related to the approach in this work.

	\subsection{Our Contribution}
	Following is a preview of our contribution.
	\begin{enumerate}
		\item We cast the joint system stability and cost minimization problem in a (stochastic) network cost minimization framework. In the online job migration model, we see that minimizing job migration cost requires keeping track of jobs' residual service requirements~(in number of time-slots).
		We introduce queue dynamics to appropriately capture this information.			
		\item We propose a job scheduling algorithm that uses the residual service requirement of the jobs in service at each scheduling epoch but does not need the arrival statistics.
		Our algorithm is based on drift plus penalty framework and, expectedly, prescribes configurations that are coupled across slots. We show that the proposed algorithm incurs a cost arbitrarily close to the optimal cost.

		\item In several scenarios, job size information may not be available on job arrivals. To address this, we propose another scheduling algorithm that merely needs to know the readily offered service to the jobs at the scheduling epochs. We show that this algorithm also achieves performance of the same order as the previous algorithm. In particular, we do not suffer~(order-wise) for not knowing jobs' sizes on their arrival.
		
		\item In the offline job migration model, we present a different system model that accounts for migration delay induced by the job migrations and present cost optimal throughput algorithms.
	\end{enumerate}	
	
	\paragraph*{Organization of the paper}
	We explain cloud computing paradigm Infrastructure as a Service~(IaaS) and different costs associated in Section~\ref{DCN-sec:IAAS}. 
	 We introduce system model for online migration cost and formulate the cost optimization problem in Section~\ref{DCN-sec:system-model}. In Section~\ref{DCN-sec:job size}, we first characterize the optimal average cost and then propose a job scheduling algorithm that achieves this cost arbitrarily closely. This algorithm requires jobs' size information on their arrival. In Section~\ref{DCN-sec:uk job size}, we propose another algorithm that achieves similar performance but merely needs to know the readily offered service to the jobs at the scheduling epochs. In Sections~\ref{DCN-sec:offline-mig-sys-model} and~\ref{DCN-sec:offline-opt-scheduling} we introduce offline migration system model and present throughput algorithms respectively. We demonstrate the performance of the proposed algorithms and compare these to the existing algorithms in Section~\ref{DCN-sec:simulations}. Finally, we conclude in Section~\ref{DCN-sec:conclusion}.
\section{CLOUD COMPUTING PLATFORMS}\label{DCN-sec:IAAS}
A cloud computing platform is a cluster of distributed computers providing on-demand computational resources or services to remote users over a network \cite{Iaas-LI}. IaaS is a cloud computing platform that provides users with computing resources such as storage, networking, and servers depending on users' requirements, e.g., Amazon EC2. Further, it can dynamically scale up the allocated resources as per users' demand. So, the users need not worry about procuring, installing, and maintenance of the resources. The users "pay as they use," i.e., they only pay for the time they have used the resources \cite{UnderstandingCC}, \cite{Buyya2009CloudCA}. The provisioning of these resources is in the form of VMs deployed on servers~(physical machines) in the data center. A VM is an abstract unit of storage and computing capacity provided in the cloud. VMs are classified according to the resources they request. For example, Table~\ref{table:Table 1} lists three types of VMs available in Amazon EC2~\cite{stheja2012loadbal}. We assume that the platform offers $M$ different types of VMs. 
\begin{table}
	\begin{center}
		\footnotesize		
		\begin{tabular}{|l | l | l | l|}
			\hline \rule{0pt}{15pt}
			Instance type & Memory &CPU & Storage\\ [0.5ex]
			\hline\hline\rule{0pt}{13pt}
			Standard Extra Large & $15$GB  & $8$EC2 units & 1690GB  \\
			\hline \rule{0pt}{13pt}
			High-Memory Extra Large & $ 17.1 $GB& $ 6.5 $EC2 units & $ 490 $ GB  \\
			\hline \rule{0pt}{13pt}
			High-CPU Extra Large & $ 7 $ GB& $ 20 $ EC2 units & $1690$ GB\\				
			\hline
		\end{tabular}
	\caption{ Three representative instances in Amazon EC2}
	\label{table:Table 1} 
	\end{center}
\vspace{-1cm}	
\end{table}
\par As the jobs arrive at the cloud, depending on their resource requirements VMs are assigned to the jobs, e.g., if a job requires more processing power then High-CPU Extra Large VM is assigned. Using resources, e.g., CPU, memory incurs power consumption at the servers, and scheduling a VM on a server incurs power consumption depending on VM's resource usage on the server. Now onwards we use the terms VM and job interchangeably. 
Below we describe the power consumption model of the servers.
\subsection{Server Power Consumption}
We use an additive server power model, where the power consumption at a server is an affine function of the number of VMs of different types that are in execution~\cite{Survey-Energyconsumption}. We explain this in the following. The total power consumption at a server is divided into two parts, static power consumption, and dynamic power consumption~\cite{Online-Li-2012}. The power consumed by the server irrespective of the number VMs scheduled on the server is called static power consumption and the power consumed by the VMs is called dynamic power consumption. The static power consumption is approximately $70\%$ of the full loaded server power consumption~\cite{Luo2014HybridSF}. Using this, the total power consumption at a server can be written as,
\begin{align}\label{DCN-eq:Power-Server}
P_{server} = P_{static}+\sum_{m=1}^Mn_mP_{VM}(m)
\end{align}
where $P_{server}, P_{static} $ and $P_{VM}(m)$ represent the total power consumption, static power consumption, power consumed by type-$m$ VM and $n_m$ is the total number of VMs of type-$m$ that are scheduled on the server. Further, power consumed by a type-$m$ VM can be written as,
\begin{align*}
P_{VM}(m) = P_{CPU}(m)+P_{mem}(m)+P_{storage}(m)
\end{align*}
where $ P_{CPU}(m), P_{mem}(m), P_{storage}(m)$ represent power consumed by the resources CPU, memory and storage.
\paragraph*{Binary server power model} For ease of exposition, in Sections~\ref{DCN-sec:uk job size} and \ref{DCN-sec:offline-opt-scheduling} we use a binary server power model wherein an ON server~(active server) consumes a fixed power irrespective of the number VMs scheduled. It is equivalent to considering only the static power consumption. In this case,
\begin{align} \label{DCN-eq:Binary-server-power}
P_{server} = 
\begin{cases}
P_{static} &\text{ if server is ON} \\
0& \text{ otherwise}
\end{cases}
\end{align}
\paragraph*{Server consolidation} To ensure service reliability and resource availability for different applications, the servers in the cloud platforms are over provisioned with the resources. However, on average $30\%$ of servers remain idle and use $10$-$15\%$ of available resources \cite{Mig-Cost-Uddin}. This underutilization of resources remarkably increases the power consumption as each server consumes $P_{static}$ irrespective of the number of jobs scheduled on it, thus increasing operational expenditure in the cloud platforms. The jobs on the underutilized servers can be consolidated onto fewer servers to improve power efficiency and resource utilization \cite{Mig-Survey}. Specifically, packing the maximum number of jobs onto a minimum number of servers, and turning off the unnecessary servers reduces power consumption in cloud platforms is called server consolidation. During server consolidation, the jobs on underutilized servers are preempted and moved to the other servers, which is referred to as job migration. job migration incurs either additional cost on the cloud or additional delay in the execution depending on the type of migration as described below. 
 \subsection{Job Migration}
The jobs are migrated from one server to another server to achieve several objectives, e.g., server consolidation, load balancing, fault tolerance, server maintenance, and isolation of applications~\cite{Mig-Singh}. We call the server on which the job was in execution before migration as the source server and the one to which it moves after migration as the destination server. During job migration, the resource status such as CPU registers content and I/O states, etc., which are referred to as VM pages should be transferred from the source server to the destination server~\cite{YANG2014267}. During the transfer of VM pages, the job may not be executed on any of the source or destination servers for a certain amount of time and we call this downtime. If the jobs are delay sensitive, larger downtime can degrade jobs' performance. There are two types of job migration depending on the modalities of the migration process and the resulting downtime. 
\subsubsection{Offline migration} 
In offline migration, the job at the source server is paused and the corresponding VM pages are transferred to the destination server. After VM pages are transferred, the job resumes its execution at the destination server. Offline migration incurs substantial downtime which ranges from a few milliseconds to seconds depending on the network bandwidth and amount of data to be moved~\cite{mig-time}.

\begin{figure}[h!]
	\begin{center}
		\includegraphics[width=3in,height=1.25in]{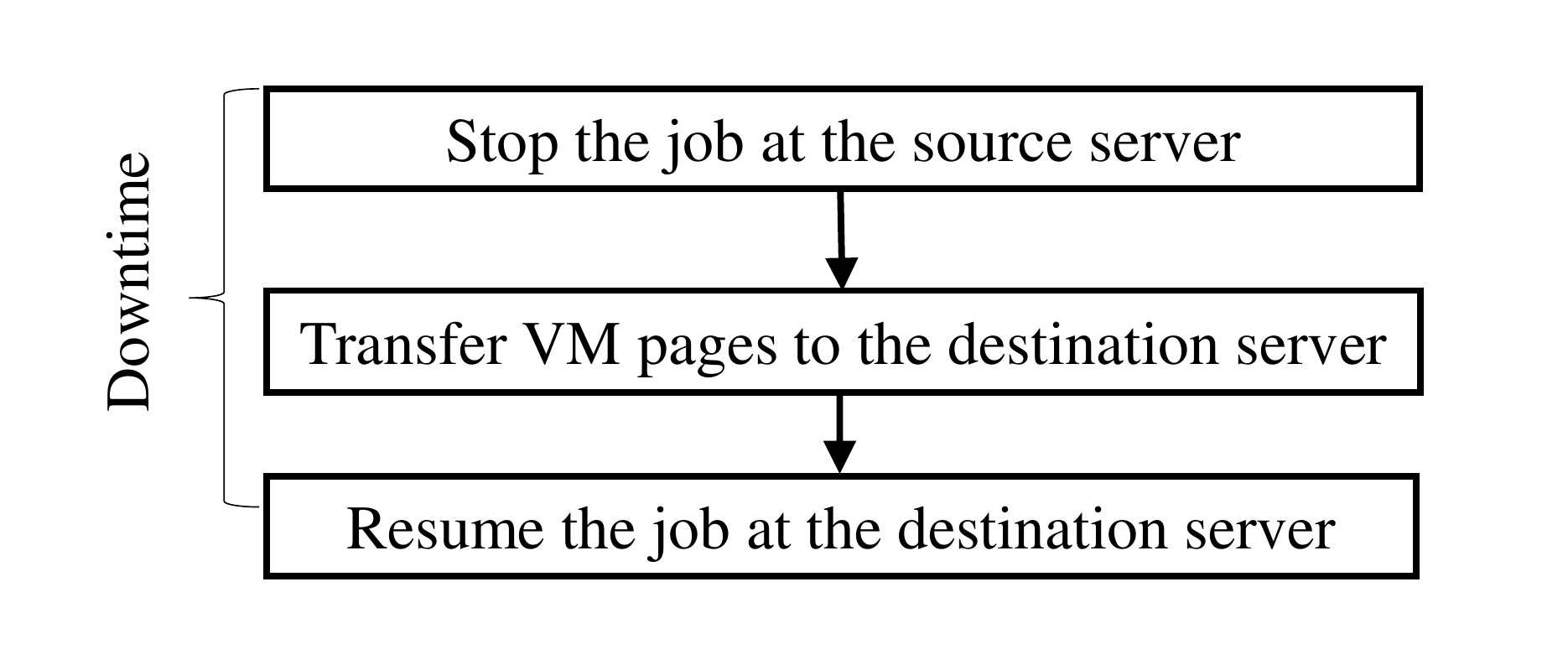}
		\caption{Basic steps that are followed in offline migration }
		\label{fig:offline migration}
	\end{center}
\end{figure}
\subsubsection{Online Migration}
Online migration leads to shorter downtime at the expense of more resource consumption. There are two types of online migrations~\cite{HOSSEIN-2020267}.
\paragraph*{Pre-copy migration} In this, the total VM pages are transferred from source server to the destination server before resuming the VM execution at the destination server. During the transfer process, if any of the pages that are already copied to the destination server are modified~(become 'dirty'), the source server has to send them again which consumes additional network bandwidth. This iterative copying continues for a limited number of rounds~\cite{HOSSEIN-2020267}. Then, the job is stalled on the source server and remaining VM pages are copied and the job execution is resumed at the destination server. In this model, the downtime is significantly small compared to offline migration as most of the VM pages are transferred before the job is moved to the destination server~(see Figure~\ref{fig:pre-copy migration}).
 \begin{figure}[h!]
 	\begin{center}
 		\includegraphics[width=3in,height=1.75in]{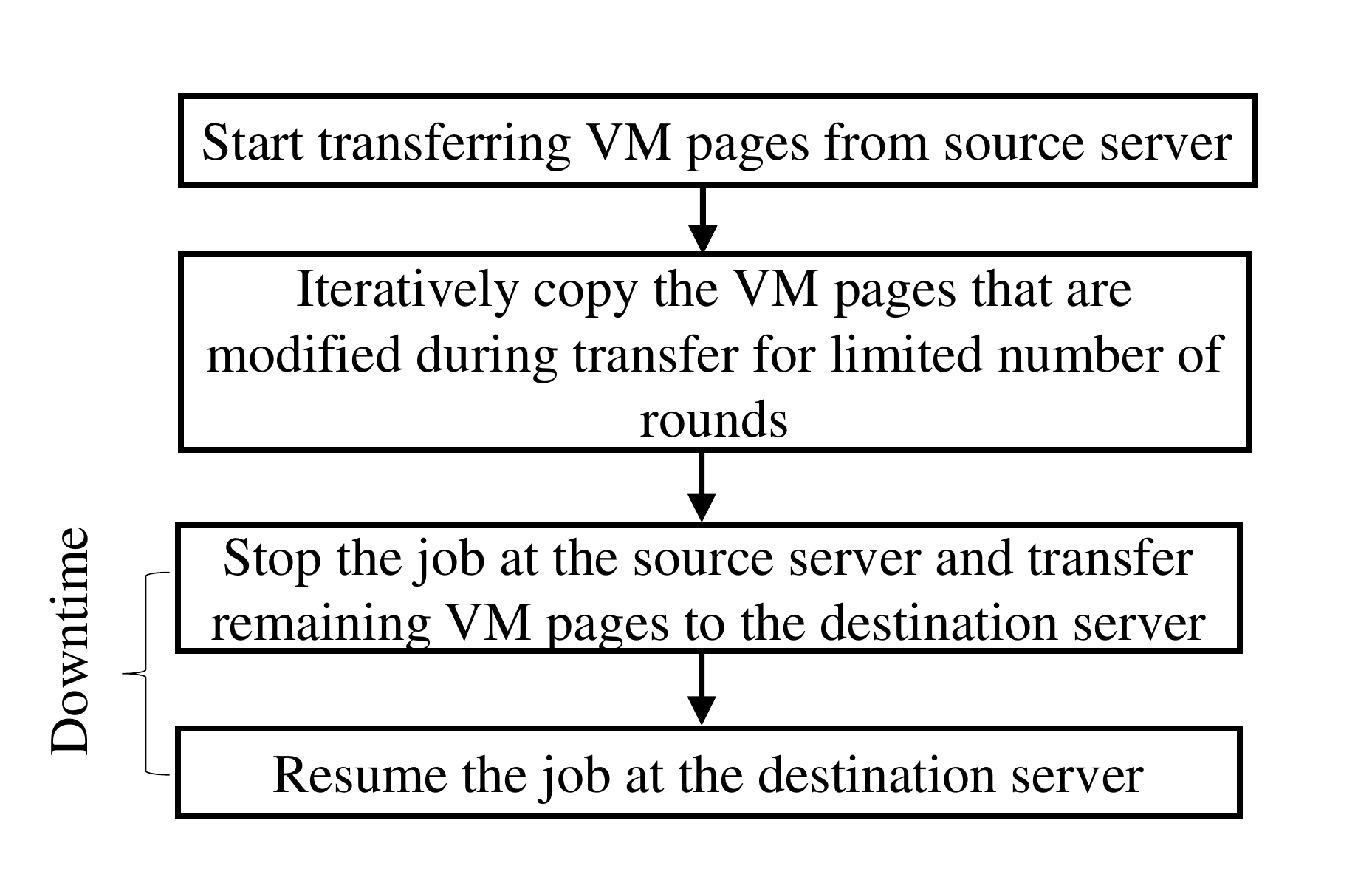}
 		\caption{Basic steps that are followed in pre-copy migration }
 		\label{fig:pre-copy migration}
 	\end{center}
 \end{figure}
\paragraph*{Post-copy migration} In this, the job at the source server is paused and the minimum state information required to resume its execution at the destination server is transferred. If the job at the destination server tries to access the information that is not yet copied, it encounters with page fault problem and too many page faults degrade the job's performance~\cite{JIN-201423}. In this case, the downtime is equal to the time taken to copy the minimum state information to the destination server.
\par Note that in online migration either the source server may have to send a VM page multiple times (in the case of pre-copy migration) or the job may encounter page-faults degrading the user experience (in the case of post-copy migration). We model the additional resource, e.g., network bandwidth consumption or the performance degradation as a fixed job migration cost (see Section~\ref{DCN-subsubsec:jmc}). Also, observe that online migration incurs less downtime than offline migration. Further, the pre-copy migration has less downtime compared to the offline migration. So, it is the preferred mode of migration in the case of delay sensitive applications. 
\par In the online migration model, we design scheduling algorithms that balance the tradeoff between server power consumption and job migration cost while ensuring QoS~(see Section~\ref{DCN-sec:system-model}). In the offline migration model, we design scheduling algorithms that optimize power consumption while ensuring QoS (see Section~\ref{DCN-sec:offline-opt-scheduling}).

\section{ONLINE MIGRATION: SYSTEM MODEL} \label{DCN-sec:system-model}
We consider a system consisting of $ L $ servers, each having a finite amount of resources (CPU, memory, storage etc.). A VM also is characterized by the amount of various resources possessed by it. We assume that the cloud computing platform offers $ M $ different types of VMs, distinguished by their resource requirements. Also, a server can host a number of VMs of different types depending on their resource requirements. To elaborate, let there be $ K $ types of resources. For each $ l \in [L] $, let server $ l$ posses $ C_{l,k} $ units of the $k$th resource.\footnote{For any $Z \in \mathbb{Z}_{++}, [Z] \coloneqq \{1,2,\cdots,Z\}$.} Further, for each $ m \in [M] $, let VM of type $ m $ require $ R_{m,k} $ units of the $k$th resource. Then server $ l$ can host $ \eta_{m} $ VMs of type $ m $, $ m \in [M]$, if
	\begin{equation*}
	\sum_{m \in [M]}\eta_mR_{m,k} \leq C_{l,k} \text{  } \forall k. 
	\end{equation*}
	We refer to $ \eta = (\eta_1 \dots , \eta_{M}) $ as a feasible VM configuration at server $ l$. Let $ \mathcal{W}^l $ be the set of all feasible VM configurations at server $ l$.
	A job is characterized by the type of VM it requires and also the duration for which the VM is required. The latter is measured in units of time-slots and is referred to as size of the job. We assume that jobs' maximum size is $S$. Also, in the following, we refer to a job as a type-$(m,s)$ job if it requires a VM of type-$m$ for $s$	slots.	
	
	\paragraph*{Job Arrival model}
	We assume that jobs arrive according to a discrete time process and i.i.d. across slots. We also assume that the number of jobs of any particular type and size arriving in a slot is bounded by $ A_{\text{max}} $. Let $ A_{m,s}(t) $ be the number of type-$(m,s) $ jobs the arrived at slot $ t $. Let $ \E[A_{m,s}(t)] = \lambda_{m,s} $, $ m \in [M],s \in [S] $.\footnote{We assume that Prob$(A_{m,s}(t)=0) > 0$ for all $ m \in [M],s \in [S] $. This ensures that the resulting queue length evolution process is irreducible and aperiodic.}
	
	\subsection{Queue Length Evolution}\label{DCN-subsec:online-mig-Q-evol}
	All the jobs arriving to the cloud computing platform are placed in queues until they finish their service. The authors in \cite{stheja2012loadbal},\cite{stheja-unknown} put all the jobs that require VMs of the same type in the same queue. They define configuration for a slot to be the number of VMs of different types offered at that slot.
	They also define a schedule to be a sequence of VM configurations, one for each slot. However, as we explain below, if we want to consider job migration cost, we require different queues for different job types, i.e., we require $MS$ queues. A type-$(m,s)$ job, if served at slot $t$, becomes a type-$(m,s-1)$ job at $t+1$. We let $(Q_{m,s}(t))_{m \in [M], s \in [S]} $ denote the queue lengths at $t$. We define a job service configuration at server $l$ as $N^l=(N^l_{m,s})_{m \in [M], s \in [S]}$, where $N^l_{m,s}$ denotes the number of VMs offered to type-$(m,s)$ jobs. We can define the set of all feasible job service configurations at server $l$ as follows,
	\begin{align} \label{DCN-eq:job-serv-config}
	\mathcal{N}^l:= \left\{ (N_{m,s})_{m \in [M],s \in [S]} : \left(\sum_{s\in [S]} N_{m,s}\right)_{m \in [M]}\in \mathcal{W}^l \right\}.
	\end{align}
	We let $N = (N^l, l \in [L])$ denote the job service configuration of the cloud computing cluster.
	Now a schedule is defined to be a sequence of job-service configurations $N(t)=(N^l_{m,s}(t)), t \geq 0$;
	$N^l_{m,s}(t)$ denotes the number of VMs offered by server $l$ to type-$(m,s)$ jobs at slot $t$. Let us also define $N^l_m(t)= \sum_sN^l_{m,s}(t)$, for all $ m \in [M] $. Further, let $N_{\max}$	 be the maximum number of VMs that any server can host. Let $\bar{N}(t)=(\bar{N}^l_{m,s}(t)), t \geq 0  $ be the job service configuration given by the scheduling algorithm. As there may not be enough jobs in the queues, the actual number of type-$(m,s)$ jobs served at server $l$ at slot $t$, say $N_{m,s}^l(t)$,
	may be less than $\bar{N}_{m,s}^l(t)$. In fact, the total number of type-$(m,s)$ jobs served at slot $t$ will be $\min\{Q_{m,s}(t), \sum_l \bar{N}_{m,s}^l(t)\}$.
	The queue lengths evolve as
	\begin{align} \label{DCN-eq:queue-evolution}
	Q_{m,s}(t+1)&=\left(Q_{m,s}(t)-\sum_l\bar{N}_{m,s}^l(t)\right)^+\nonumber\\
	+A_{m,s}(t)&+\min\left\{\sum_l\bar{N}_{m,s+1}^l(t), Q_{m,s+1}(t)\right\}.
	\end{align}
	The cloud system is called stable if $\lim_{t\to \infty}\sum_{m,s}\E[Q_{m,s}(t)] < \infty.$	
	\subsection{System Capacity} \label{DCN-sec:system cap}
The capacity region of the cloud computing cluster is defined to be the set of job arrival rates for which there exists a scheduling algorithm 
that renders the system stable. The capacity region of the cloud computing cluster under consideration is well known~\cite{stheja2012loadbal, stheja-unknown}. Following is the capacity region is given by~\cite{stheja2012loadbal}, but expressed in terms of job arrival rates.

\begin{align} 
	\Lambda= \Bigg\{\lambda&:  \text{ for all } l \in [L] \text{ there exist } \beta^l_{W^l}, W^l \in \mathcal{W}^l, \nonumber\\
	&\text{ such that } \beta^l_{W^l} \geq 0, \sum_{W^l \in \mathcal{W}^l}\beta^l_{W^l} =1 \nonumber \\
	&  \text{ and } \sum_s s\lambda_{m,s} = \sum_{l}\sum_{W^l \in {\cal W}^l} \beta^l_{W^l}W^l_m, \forall m\in [M]  \Bigg\}.\nonumber 	
\end{align}
A scheduling algorithm is said to be throughput optimal if it can stabilize network queues for any job arrival rates $ \lambda $ such that $\lambda + \e \in \Lambda $ for some $ \epsilon >0 $.\footnote{With slight abuse of notation, we write $\lambda+1\e=\lambda+\e$.}
	\subsection{Costs} 
	We now formalize the costs that we aim to optimize.
	
\subsubsection{Server Running Cost} \label{DCN-subsubsec:src} 
Recall that each scheduled VM on a server consumes a certain amount of resources, using these resources incurs power consumption. Motivated by \cite{Online-Li-2012}, we adopt an affine cost structure to model the power consumption as discussed in Section~\ref{DCN-sec:IAAS} (see \eqref{DCN-eq:Power-Server}). Let $ \zeta_k $ be the power consumption cost per unit slot incurred for using a resource of type-$k$ on a server. The total cost incurred for executing type-$m$ VM on the server is given per one slot by
\begin{align}
c_m = \sum_kR_{m,k}\zeta_k.
\end{align}
Let $c_0$ be the static power cost incurred by the server.
Using this, we write the total power consumption cost in one slot incurred by the server-$l$ to host a job service configuration $N \in \mathcal{N}^l$ as

\begin{align}\label{DCN-eq:generic_src}
C_1^l(N)= & c_0+\sum_m c_m \Big(\sum_sN_{m,s}\Big). 
\end{align}

\begin{remark}\label{rem:bin-server}
	For binary server power model as described in Section~\ref{DCN-sec:IAAS}, the power consumption at server $l$ is given by $C_1^l(N) = c_0$.
\end{remark}
We instead consider an upper bound on this cost, $C^l_1(\bar{N}^l)$, in our objective function for ease of analysis. Please see the remark at the end of this section.

\subsubsection{Online Job Migration Cost} \label{DCN-subsubsec:jmc} 
Recall that in the online migration model the migrated job suffers less downtime compared to the offline migration model and incurs job migration cost. In our analysis, we assume the downtime is equal to zero and model the network bandwidth consumption~(in the case of pre-copy migration) and job's performance degradation due to page faults~(in the case of post-copy migration) as job migration cost. We formalize job migration cost as follows.
\par When a scheduling algorithm prescribes a new job service configuration $ \bar{N}(t) $ at slot $ t $, certain jobs may need to be preempted and their states have to be stored. To illustrate this let us assume that a server $l$ serves a VM of type-$m$
under the configuration $N(t-1)$ but does not serve any VM of type-$m$ under $N(t)$. This can happen, for example, if server $l$ is under-utilized and one wants to turn it off at $t$ to save on the server running cost. If the served VM at $t-1$ results in the completion of the corresponding job, the VM's state need not be stored. On the other hand, if the job is not yet complete, it is preempted at $t$ and its state has to be stored. When the preempted VMs resume service, their states are reloaded into the corresponding servers.
All this incurs additional cost, referred to as job migration cost. We assume that each preempted job incurs unit cost.
Notice that only those jobs that are served at $t-1$ and require at least two more slots~(at least one more
slot at $t$) can contribute to the job migration cost.
Mathematically, the total migration cost due to the jobs preempted at server $ l$ at slot $ t $ is $C_2^l(N^l(t-1),N^l(t))$ where,	
	\begin{equation}\label{DCN-eq:Migration cost}
	C_2^l(N^l,\bar{N}^l)= \sum_{m,s\geq 2} (N_{m,s}^l- \bar{N}_{m,s-1}^l)^+.
	\end{equation}
To know whether a job will get completed at $t$ or will require at least one more slot of service, one needs to keep track of its residual service requirement throughout its service. We thus see that, if we want to minimize the job migration cost, we need to keep track of the residual service requirement of the jobs. Recall that the actual number of type-$(m,s)$ jobs served at various servers at slot $t$ satisfy
	\begin{align*}
	N_{m,s}^l(t) &\leq \bar{N}_{m,s}^l(t), \ \forall l,\\
	\sum_l N_{m,s}^l(t) &= \min\left\{Q_{m,s}(t),\sum_l\bar{N}_{m,s}^l(t)\right\}.
	\end{align*}
We impose the following additional condition
\begin{equation} \label{DCN-eq:cond-1}
	N_{m,s}^l(t) \geq \min\left\{N_{m,s+1}^l(t-1), \bar{N}_{m,s}^l(t)\right\}.
\end{equation}
	We can justify this restriction as follows. Server $l$ keeps $\min\{N_{m,s+1}^l(t-1), \bar{N}_{m,s}^l(t)\}  $ of type-$ (m,s) $ jobs from slot $ t-1 $ at slot $ t $. In other words, only $ \left(N^l_{m,s+1}(t-1)-N^l_{m,s}(t)\right)^+$ of these jobs are migrated. This is a reasonable assumption since, given $ N(t-1) $ and $ \bar{N}(t) $, any other strategy would incur strictly larger job migration cost. Observe that Algorithm \ref{Algo:Size} (see Step 2), meets this requirement.
	A job service configuration meeting \eqref{DCN-eq:cond-1} for all types of job is possible because~(see~\eqref{DCN-eq:queue-evolution})
	\[Q_{m,s}(t) \geq \sum_lN_{m,s+1}^l(t-1).\]
	We now argue that, under a sequence of job service configurations, $N(t)$ satisfying \eqref{DCN-eq:cond-1}, the job migration cost $C_2^l(N^l(t-1),N^l(t))$ equals 	
	$C_2^l(N^l(t-1),\bar{N}^l(t))$.	To see this, consider the following two cases.	
			
	\begin{enumerate}

		\item $ N_{m,s-1}^l(t) \geq N_{m,s}^l(t-1) $: In this case,
		\begin{align*}
		&(N_{m,s}^l(t-1)-\bar{N}_{m,s-1}^l(t))^+ \\
		&= (N_{m,s}^l(t-1)-N_{m,s-1}^l(t))^+ =0.
		\end{align*}
		\item $ N_{m,s-1}^l(t) < N_{m,s}^l(t-1) $:	In this case, $\newline N_{m,s-1}^l(t) = \bar{N}_{m,s-1}^l(t)$, and so
		\begin{align*}
		(&N_{m,s}^l(t-1)-\bar{N}_{m,s-1}^l(t))^+ \\
		&= (N_{m,s}^l(t-1)-N_{m,s-1}^l(t))^+.
		\end{align*}		
	\end{enumerate}
	So we see that $C_2^l(N^l(t-1),N^l(t)) =  C_2^l(N^l(t-1),\bar{N}^l(t))$. 
\par We aim to minimize the time average value of the linear combination of $C_1^l(\bar{N}^l(t))$ and $C_2^l(N^l(t-1),\bar{N}^l(t))$.
\begin{equation}\label{DCN-eq:cost-online}
	h(t)=V\sum_lC_1^l(\bar{N}^l(t))+ U\sum_{m,s\geq 2}C_2^l(N^l(t-1),\bar{N}^l(t)),
	\end{equation}
where $ U$ and $V $ are non-negative weighing parameters that can be tuned to yield a tradeoff between
delay, the server running cost and the job migration cost.
Formally, the optimization problem can be posed as
	\begin{align*} 
	&\min \lim_{T \to \infty}\frac{1}{T}\sum_{t=0}^{T-1}\E[h(t)]  \\
	& \text{subject to   }\lim_{t\to \infty}\sum_{m,s}\E[Q_{m,s}(t)] < \infty,  \\
	& (\bar{N}_{m,s}^l(t))_{m \in [M], s \in [S]}\in \mathcal{N}^l, \text{  }\forall l \in [L].
	\end{align*}
	\begin{remark}
		1)~If job arrival rates are close to Pareto boundary of the capacity region, in most of the slots there are enough jobs in the system to ensure $ \bar{N}^l(t) = N^l(t) $. In this case $ \sum_lC_1^l(\bar{N}(t)) $ is a tight upper bound on the actual server running cost.\\
		2)~In Section~\ref{DCN-sec:job size}, we provide a job scheduling algorithm~(Algorithm~\ref{Algo:Size}), which can make $\lim_{T \to \infty}\frac{1}{T}\sum_{t=0}^T \sum_l \E[C_1^l(\bar{N}(t))]$ arbitrarily close to $\bar{C}_{\text{opt}}(\lambda)$~(see~\eqref{DCN-eq:Opt-2}). The latter is also the least possible value of time averaged actual server running cost, i.e.,
		$\lim_{T \to \infty}\frac{1}{T}\sum_{t=0}^T \sum_l \E[C_1^l(N(t))]$. Thus Algorithm~\ref{Algo:Size} can also yield actual server running cost arbitrarily close to its optimal value.
	\end{remark}
	\begin{center}
	\small
		\begin{tabular}{||l  l||}
			\hline \rule{0pt}{15pt}
			Symbol & Description  \\ [0.5ex]
			\hline\hline\rule{0pt}{13pt}
			$ L $ & Number of servers  \\
			$ M $ & Number of types of jobs \\
			$ S $ & Maximum size of the job \\
			$ \Lambda $&Capacity region of the cloud system  \\
			$ \lambda $ &Job arrival rate vector\\
			$ {\cal N}^l $ & Set of job service configurations at server $l$;\\
			& each job service configuration is a $ M \times S $
			matrix\\
			  $\bar{N}^l(t)$ &Job Service configuration at server $l$ \\
			&  at slot $t$ under Algorithm~\ref{Algo:Size}\\ 
			$\tilde{N}^l(t)$ &Job service configuration at server $l$  \\
			& at slot $t$ under Algorithm~\ref{Algo:UK-algo-1}\\
			$N^l(t)$ &Actual job service configuration at server $l$\\ & at slot $t$ \\			
			$Z^l(t)$ &Completed jobs at server $l$ 
			at the end of slot $t$ \\			
			\hline
		\end{tabular}

	\end{center}

	\section{ONLINE MIGRATION: SCHEDULING WITH KNOWN JOB SIZES} \label{DCN-sec:job size}
	We start with characterizing the minimum average cost required
	to keep the system stable. Towards this we consider
	the class of all scheduling policies, including those with
	complete knowledge of job arrival statistics, past job arrivals, past job service configurations and also future arrivals. For
	any $l \in [L]$, we let ${\cal P}^l$ be the set
	of all probability distributions over ${\cal N}^l$
	and ${\cal S}^l$ the set of all stochastic matrices in $\mathbb{R}^{{\cal N}^l \times {\cal N}^l}$.
	Let ${\cal S} \coloneqq \prod_{l \in [L]}{\cal S}^l$ and ${\cal P} \coloneqq \prod_{l \in [L]}{\cal P}^l$.
	Further, $\pi = (\pi^l)_{l \in [L]}$
	and $P = (P^l)_{l \in [L]}$
	denote vectors of probability distributions and
	stochastic matrices, respectively, for all $l \in [L]$; $\pi \in {\cal P}, P \in {\cal S} $.
	Any stabilizing scheduling policy incurs an average cost
	greater than $C_{\text{opt}}(\lambda)$, the optimal value
	of the following problem:
	\begin{align*}
	\min_{(\lambda^l)_{l \in [L]},P \in {\cal S},\pi
		\in {\cal P}} & \sum_l \left(V C^l_{\pi^l,1}(\lambda^l) + U C^l_{P^l,\pi^l,2}(\lambda^l)\right), \\
	&\text{such that} \sum_l \lambda^l = \lambda, \nonumber \\
	\pi^l P^l &= \pi^l, l \in [L], \nonumber\\
	\sum_s s\lambda_{m,s}^l &\leq \sum_{N^l \in {\cal N}^l}\pi^l_{N^l}\sum_s(N^l_{m,s}), \forall l,m,\nonumber
	\end{align*}
	where
	\begin{align} 
	C^l_{P^l,\pi^l,2}&(\lambda^l) \coloneqq  \sum_{N^l \in {\cal N}^l}\pi^l_{N^l} \nonumber \\ 
	&\Bigg(\sum_{N'^l \in {\cal N}^l}P^l_{N^l,N'^l}\sum_{m,s\geq 2} \left[N^l_{m,s}-N'^l_{m,s-1}\right]^+\Bigg),\label{DCN-eq:cost2} \\
	C^l_{\pi^l,1}&(\lambda^l) \coloneqq  \sum_{N^l \in {\cal N}^l}\pi^l_{N^l}\Bigg(c_0\mathbbm{1}_{N^l\ne 0}+\sum_mc_m\bigg(\sum_sN_{m,s}\bigg)\Bigg)\label{DCN-eq:cost1}
	\end{align}
	\begin{lemma}\label{Lemma:Markov-The-Best}
	For any arrival rate $\lambda \in \Lambda$ and any stabilizing scheduling policy, the average cost is lower bounded by $C_{\text{opt}}(\lambda)$.
	\end{lemma}
	\begin{proof}
	See Appendix~\ref{Proof:Markovian the Best}.
	\end{proof}

	We show that our proposed algorithm~(Algorithm~\ref{Algo:Size}) stabilizes the system and achieves a cost arbitrarily close to $ C_{\text{opt}}(\lambda) $. 
	In fact, Algorithm~1 achieves a cost arbitrarily close to $ \bar{C}_{\text{opt}}(\lambda) $, the optimal value of the following
	problem:
	\begin{align}
	\min_{(\lambda^l)_{l \in [L]}, \pi \in {\cal P}} & \sum_l C^l_{\pi^l,1}(\lambda^l) \label{DCN-eq:Opt-2} \\
	\text{such that} & \sum_l \lambda^l = \lambda, \nonumber \\
	&  \sum_s s\lambda_{m,s}^l \leq \sum_{N^l \in {\cal N}^l}\pi^l_{N^l}\sum_s N^l_{m,s}. \nonumber
	\end{align}
	
	Since $\bar{C}_{\text{opt}}(\lambda)$ lower bounds $C_{\text{opt}}(\lambda)$, our claim holds. Furthermore, consider a $\lambda$ such that $\lambda + \epsilon \in \Lambda$ for some $\epsilon > 0$. Let an optimal solution of the above problem, for arrival rates $\lambda +\epsilon$, be attained at $\bar{\pi} \in \cal{P}$. Then, for arrival rates $\lambda$, 
(a) The static scheduling policy $\bar{\pi}$ stabilizes the system,  
(b) it incurs a server running cost upper bounded by $\bar{C}_{\text{opt}}(\lambda+\e)$.
Finally, $\bar{C}(\cdot)$~(see \cite[Lemma~1]{subhashini2017augment}) is a continuous function of arrival rates. Hence, by choosing $\epsilon > 0$ arbitrarily small, we can obtain static policies that stabilizes the system and incur server running cost arbitrarily close to $\bar{C}_{\text{opt}}(\lambda)$.
\subsection{Drift plus Penalty based Algorithm}
We propose a job scheduling algorithm, Algorithm~\ref{Algo:Size}, parameterized by non-negative weighing parameters $U $ and $V$. 
\begin{algorithm}
		\caption{Online migration-Known job sizes}
		\label{Algo:Size}
		\begin{algorithmic}[1]
			\State Given $ N(t-1)=n , Q(t)=q $, for all $l \in [L]$,
			 \begin{align} \label{DCN-eq:drift+penalty-control-action}
			\bar{\mathcal{N}}^{l}(t) = \argmax_{N \in \mathcal{N}^{l}}& \Bigg[\sum_m \Big(\sum_ssq_{m,s}\Big) \Big(\sum_s N_{m,s}\Big)\nonumber\\
				- V&C_1^l(N)  -UC_2^l(n,N)\Bigg],\\
			\bar{N}^l(t) \in \argmax_{N \in \bar{\mathcal{N}}^{l}(t)}&\sum_{m,s}q_{m,s}N_{m,s},\nonumber
			\end{align}
			\textbf{Obtaining $ N(t) $ from $ \bar{N}(t)$:}
			\State For all $l \in [L], m\in[M],s \in [S]$,
			\[\hat{N}^l_{m,s}(t) = \min\{\bar{N}^l_{m,s}(t),n^l_{m,s+1}\},\]
			\State For all $m\in[M],s \in [S]$,
			\[q_{m,s} = q_{m,s}- \sum_l\hat{N}^l_{m,s}(t),\]
			\State For all $m\in[M],s \in [S], l \in [L]$, 
			\begin{align*}
			N^l_{m,s}(t) = &\min\{q_{m,s},\bar{N}^l_{m,}(t)-\hat{N}^l_{m,s}(t)\}+\hat{N}^l_{m,s}(t) \nonumber \\
			q_{m,s} = &q_{m,s}-\min\{q_{m,s},\bar{N}^l_{m,}(t)-\hat{N}^l_{m,s}(t)\}
			\end{align*}
		\end{algorithmic}
\end{algorithm}
In Step~$1$, when $ \bar{\cal{N}}^l(t) $ is not a singleton, different job service configurations in $ \bar{\cal{N}}^l(t) $ may correspond to different number of jobs being served. Selecting $ \bar{N}(t) $ as in this step ensures that $ \min\{\sum_s \bar{N}_{m,s}(t), \sum_s q_{m,s}\} $ type-$(m,s)$ jobs are served. We use this property in the proof of Theorem \ref{Theorem:Size-Algo-Optimality}. Steps 2-4 ensure that the job service allocations at all the servers meet condition \eqref{DCN-eq:cond-1}. We have provided the pseudo code for Algoritthm~\ref{Algo:Size} in Appendix~\ref{Appenidix:Pseudocode}.
The following theorem establishes the optimality of the above algorithm.
\begin{theorem}\label{Theorem:Size-Algo-Optimality}
For any arrival vector $ \lambda $ such that $ \lambda + \e \in \Lambda$ for some $\epsilon >0$, Algorithm \ref{Algo:Size} achieves average queue length and average cost bounds as given below.
\begin{enumerate}[(a)]
\item	 
$\lim_{T\to \infty}\frac{1}{T}\sum_{t=0}^{T-1}\sum_{m,s}\E[sQ_{m,s}(t)]  \newline \newline$
$  \leq \frac{1}{\epsilon}\Big(B_{2}+ ULN_{\max}+V\bar{C}_{\text{opt}}(\lambda+\e)\Big),\newline$ 
	
	where $B_2 = \frac{MS}{2}\Big(SA^2_{\max}+2(LN_{\max})^2\Big)\newline$
\item
$\lim_{T\to \infty}\frac{1}{T}\sum_{t=0}^{T-1}\sum_l\E[C^l_1(\bar{N}^l(t))]  \newline\newline
	\textrm{    } \leq \dfrac{ B_2+ULN_{\max}}{V}+\bar{C}_{\text{opt}}(\lambda+\e), \newline$
\item
$\lim_{T\to \infty}\frac{1}{T}\sum_{t=0}^{T-1}\sum_l\E\Big[C^l_2(N^l(t-1),\bar{N}^l(t))\Big]	\leq \dfrac{B_4}{U},\newline\newline$
where $B_4 = LMSA_{\max}N_{\max}+LMN^2_{\max}.\newline$
\end{enumerate}
\end{theorem} 
	\begin{remark}
		By using large enough $U$ and $V $, the algorithm can achieve a cost close to the optimum cost $C_{\text{opt}}(\lambda)$.
		However, large $U$ and $V$ tend to give large delays.
	\end{remark}
	\begin{proof}
		See Appendix~\ref{Proof:Size-Algo-Optimality}.
	\end{proof}

\section{ONLINE MIGRATION: SCHEDULING WITH UNKNOWN JOB SIZES} \label{DCN-sec:uk job size}	
	In several scenarios, job sizes may not be available on job arrivals. To address this, we propose an optimal cost scheduling algorithm that distinguishes jobs based on readily offered service to them. More specifically, now $(\tilde{Q}_{m,a}(t))_{m \in [M], 0 \leq a\leq S-1}$ denote the number of jobs in the system, that require VMs of type-$m$ and have been offered $a$ slots of service until time $t$. We now refer to these jobs as type-$(m,a)$ jobs. In particular, all external arrivals of type-$m$ at slot $t$ join $\tilde{Q}_{m,0}(t)$. A job service configuration is now expressed as $N=(N_{m,a})_{m \in [M], 0 \leq a\leq S-1}$, where $N_{m,a}$ represents the number of VMs offered to type-$(m,a)$ jobs. 
	\subsection{Queue Length Evolution}
	Let $\tilde{N}(t), t \geq 0$ be the job service configuration prescribed by certain scheduling algorithm. As explained earlier, the actual job service configuration at slot $t$ can be different from $\tilde{N}(t)$. Let us use $N(t)$ to denote the former. 
	Let $Z(t) = (Z_{m,a}(t))_{m \in [M], 0 \leq a \leq S-2}$, where $Z_{m,a}(t)$ is the number of type-$(m,a)$ jobs that complete their service at $(t+1)-$. We can then write queue length evolution as follows.
	\begin{align*}	
	\tilde{Q}_{m,0}(t+1)=&\left(\tilde{Q}_{m,0}(t)-\sum_l\tilde{N}^l_{m,0}(t)\right)^++\sum_s A_{m,s}(t), \\
	\tilde{Q}_{m,a}(t+1)=&\left(\tilde{Q}_{m,a}(t)-\sum_l\tilde{N}^l_{m,a}(t)\right)^+ \\
	+\min &\left\{\sum_l\tilde{N}^l_{m,a-1}(t), \tilde{Q}_{m,a-1}(t)\right\}-Z_{m,a-1}(t).
\end{align*}	
As in Section~\ref{DCN-subsubsec:src} the server running cost is $C^l_1(N^l(t))$. For ease of analysis we consider only static power consumption cost, with $c_0=1, c_1=\dots=c_m = 0$ as the server running cost. But, the analysis holds for the affine cost structure. The job migration cost can be defined as follows, 
\begin{align*}
\tilde{C}^l_2&\Big(N^l(t-1),Z^l(t-1),\tilde{N}^l(t)\Big)\\
& = \sum_m\sum_{a=0}^{S-2}\left(N_{m,a}^l(t-1)-Z_{m,a}^l(t-1)-\tilde{N}_{m,a+1}^l(t)\right)^+.
\end{align*}
Using similar arguments as in Section \ref{DCN-subsubsec:jmc} it can be shown that
\begin{align*}
\tilde{C}^l_2&\Big(N_{m,a}^l(t-1),Z_{m,a}^l(t-1),\tilde{N}_{m,a+1}^l(t)\Big)\\
=&\tilde{C}^l_2\Big(\tilde{N}_{m,a}^l(t-1),Z_{m,a}^l(t-1),\tilde{N}_{m,a+1}^l(t)\Big).
\end{align*}
	We consider the following upper bound on the total incurred cost at slot $ t $
	\begin{align}\label{DCN-eq:UK-server-running-cost}
	h&(t)=\sum_{l=1}^L \Bigg[VC^l_1(\tilde{N}^l(t))\nonumber \\
	&+U\tilde{C}^l_2\Big(N_{m,a}^l(t-1),Z_{m,a}^l(t-1),\tilde{N}_{m,a+1}^l(t)\Big) \Bigg].
	\end{align}

	\subsection{Scheduling Algorithm for Unknown Job Sizes}
	\label{DCN-sec:sched-algo2}
	We now propose a job scheduling algorithm, Algorithm~\ref{Algo:UK-algo-1}, for unknown job sizes, with the objective of 
	minimizing the time average of the cost $ h(t) $ subject to the stability of the queues
	$(\tilde{Q}_{m,a}(t))_{m\in [M], 0 \leq a \leq S-1}$. We use a logarithmic weight function, $g(q) = \log(1+q)$ to determine the weight of a job service configuration~(Step~1). Notice that we use lengths of job-queues in Algorithm~\ref{Algo:UK-algo-1} whereas its stability analysis uses backlogged workload. Using logarithmic weight function ensures that this discrepancy introduces a bounded error that has little effect at higher queue lengths.
	
\begin{algorithm}
\caption{Online migration-Unknown job sizes }
\label{Algo:UK-algo-1}
\begin{algorithmic}[1]
	\State Given $N(t-1)=n, Z(t-1) = z, \tilde{Q}(t)=q$, for all $l \in [L]$,  
	\begin{align*} 
	\tilde{\mathcal{N}}^l(t)= &\argmax_{N \in\mathcal{\bar{N}}^l}  \Bigg[\sum_m g\Big(\sum_aq_{m,a}\Big)\Big(\sum_a N_{m,a}^l\Big) \nonumber \\
	&- V\textbf{1}_{(\sum_{m,a}N_{m,a}^l>0)}  \nonumber \\ 
	&-U \sum_{m,a \geq 1} \left(n_{m,a-1}^l-z_{m,a-1}^l-N_{m,a}^l\right)^+\Bigg],\\
	\tilde{N}^l(t) \in & \argmax_{N \in \tilde{\mathcal{N}}^l(t)}\sum_{m,a}q_{m,a}N_{m,a}
	\end{align*}	
	\textbf{Obtaining $ N(t) $ from $ \tilde{N}(t)$:}
	\State For all $l \in [L], m\in[M],a \in [S-1]$,
	\[\hat{N}^l_{m,a}(t) = \min\{\tilde{N}^l_{m,a}(t),n^l_{m,a-1}-z^l_{m,a-1}\},\]
	\State For all $m\in[M],a \in \{0, \dots S-1\}$,
	\[q_{m,a} = q_{m,a}- \sum_l\hat{N}^l_{m,a}(t),\]
	\State For all $m\in[M],a \in \{0, \dots S-1\}, l \in [L]$,	
		\begin{align*}
			&N^l_{m,a}(t) \\
			&= \min\{q_{m,a}-z_{m,a-1},\tilde{N}^l_{m,a}(t)-\hat{N}^l_{m,a}(t)\}+\hat{N}^l_{m,a}(t) \nonumber \\
			&q_{m,a} = q_{m,a}-\min\{q_{m,a}-z_{m,a-1},\tilde{N}^l_{m,a}(t)-\hat{N}^l_{m,a}(t)\}
			\end{align*}

	\end{algorithmic}
\end{algorithm}
The following theorem establishes its optimality.	
\begin{theorem}\label{Theorem:UK-algo-1-optimality}
For any arrival vector $ \lambda $ such that $ \lambda + \e\in \Lambda$ for some $\epsilon >0$, Algorithm~\ref{Algo:UK-algo-1} achieves average queue length and average cost bounds as given below.
\begin{enumerate}[(a)]
	\item
	$\lim_{T\to \infty}\frac{1}{T}\sum_{t=1}^T\E[\sum_{m,a}\tilde{Q}_{m,a}(t)] < \infty,  \newline$

	\item 
	$\lim_{T \to \infty} \sum_{t=1}^{T}\sum_l \E\big[C^l_1(\tilde{N}^l(t))\big] \newline \newline \textrm{      }\leq \dfrac{K_2+LUN_{\max}}{V}+\bar{C}_{\text{opt}}(\lambda+\e_g),\newline$
where $K_2=LM(A_{\max}\bar{S}_{\max}+N_{\max}\bar{S}_{\max})+LN_{\max}g(SLN_{\max})+M\log S.\newline$
	\item 
$	\frac{1}{T}\sum_{t=1}^{T} \E\Big[\tilde{C}^l_2\Big(N_{m,a}^l(t-1),Z_{m,a}^l(t-1),\tilde{N}_{m,a+1}^l(t)\Big)\Big]\newline
	\leq  \frac{1}{U}\Bigg(MN_{\max}A_{\max}+MN_{\max}\log\Big(1+N_{\max}\Big)\Bigg) \newline$

\end{enumerate}

\end{theorem}
	\begin{proof}
			See Appendix~\ref{Proof:Age-Algo-Optimality}.
	\end{proof}	
\section{OFFLINE MIGRATION: SYSTEM MODEL}\label{DCN-sec:offline-mig-sys-model}
In this section, we describe the offline migration model. We assume that job size information is known at the time of job arrival. Further, the job arrival and server running cost models are the same as in Section~\ref{DCN-sec:system-model}. For ease of analysis we consider only static power consumption cost but, the analysis holds for the affine cost structure.
\subsection{Migration delay}\label{DCN-subsec:mig-delay-sys-model}
Recall that, in offline migration the preempted job experiences downtime and no migration cost. Further, the downtime depends on the network bandwidth and the amount of data to be copied. For ease of analysis we assume that a preempted job suffers one slot duration of downtime~(migration delay), i.e., the preempted job starts its execution at the destination server after one slot duration. In our analysis, we use an equivalent model for the migration delay caused by the job preemptions as explained in the following. Instead of starting the execution of the preempted job after migration delay~(one slot), we increase the preempted job execution time by one slot and assume that the job can start execution immediately (zero migration delay) on the destination server. For example, consider a job type-$(m,s)$ scheduled on server $l$ at time $t$, and preempted at $t+1$ and moved to the sever $l'$. At slot $t+1$ the job type is $(m,s-1)$. Since, the migration delay is one slot duration the job execution~(type-$(m,s-1)$) will start at slot $t+2$ on the server $l'$. But, in our model we assume that the job type~$(m,s)$ started execution at $t+1$ at the destination server $l'$ which is equivalent to starting the job type-$(m,s-1)$ at $t+2$.
\subsection{Queue Length Evolution}
Using the migration delay model explained in Section~\ref{DCN-subsec:mig-delay-sys-model}, we write the queue evolution,
\begin{align}\label{DCN-eq:mig-dely-q-evol}
Q_{m,s}&(t+1) \nonumber \\
=& Q_{m,s}(t)-\sum_lN^l_{m,s}(t)+A_{m,s}(t)+\sum_lN^l_{m,s+1}(t)\nonumber \\&+\sum_l\Big(N^l_{m,s}(t-1)-N^l_{m,s-1}(t)\Big)^+ \nonumber \\
&-\sum_l\Big(N^l_{m,s+1}(t-1)-N^l_{m,s}(t)\Big)^+ \text{ for } 1<s\leq S\nonumber \\ 
Q_{m,1}&(t+1) \nonumber \\
=& Q_{m,1}(t)-\sum_lN^l_{m,1}(t)-\sum_l\Big(N^l_{m,2}(t-1)-N^l_{m,1}(t)\Big)^+\nonumber \\
&+A_{m,1}(t)
\end{align}
where $N^l_{m,s}(t)$ is the job service configuration at server $l$ at slot $t$, and $\Big(N^l_{m,s}(t-1)-N^l_{m,s-1}(t)\Big)^+$ are the number of type-$(m,s-1)$ jobs that are migrated from server $l$ at slot $t$. Note that, the queue evolution is different from Section~\ref{DCN-subsec:online-mig-Q-evol}. From \eqref{DCN-eq:mig-dely-q-evol}, we write the workload evolution for VM type-$m$ as
\begin{align}\label{DCN-eq:mig-dely-wl-evol}
\sum_ssQ_{m,s}(t+1)=&\sum_ssQ_{m,s}(t)-\sum_{l,s}N^l_{m,s}(t)+\sum_ssA_{m,s}(t)\nonumber \\&+\sum_{l,s>1}\big(N^l_{m,s}(t-1)-N^l_{m,s-1}(t)\big)^+ ,
\end{align}
where $\sum_{l,s>1}\big(N^l_{m,s}(t-1)-N^l_{m,s-1}(t)\big)^+$ is the extra workload due to the type-$m$ VM migrations. Observe that the extra workload accounts for the job migration delays, so does not appear in the online migration model. We aim to minimize time average value of \eqref{DCN-eq:cost-online} with $U=0$.

\section{OFFLINE MIGRATION: SCHEDULING} \label{DCN-sec:offline-opt-scheduling}
We present two algorithms, Queue-length based Biased Max-Weight (Q-BMW) and its refined version as Algorithm 4. Q-BMW is motivated by the algorithms in~\cite{Ping2017-DelayOpt} and \cite{chang2018DCNC}. It is provably throughput optimal and yields optimal server running cost. Recall that optimizing server running cost necessitates server consolidation and so job migration. On the other hand, job migrations in the offline migration model incur downtime~(migration delay) reflected as extra workload~(see \eqref{DCN-eq:mig-dely-wl-evol}). So, optimal 
scheduling policies for offline migration must discourage frequent migrations to ensure stability, e.g., there should be fewer and fewer migrations as the arrival rates approach the Pareto boundary of the capacity region.

\subsection{Q-BMW}
In the offline migration model, when a scheduling algorithm prescribes a job service configuration that causes job migrations, it adds extra workload to the system, i.e., apart from the workload of job execution additional workload accounting for the tasks related to job migration added to the queues~(see \eqref{DCN-eq:mig-dely-wl-evol}). So, the scheduling algorithms should discourage prescribing job service configurations that cause job migrations. The idea behind Q-BMW is similar to that in \cite{Ping2017-DelayOpt} and \cite{chang2018DCNC}. Q-BMW adds a bias towards the job service configurations that do not cause job migrations. In other words, Q-BMW prescribes a job service configuration that causes job migrations only when the weight of this configuration exceeds the weights of all the other job service configurations that do not cause migrations by enough margin. 
 
\par We now describe Q-BMW. Let $\mathcal{H}^l(t)$ be the set of all job service configurations for server $l$ that do not cause job migration at this server at slot $t$. More precisely,
\begin{align*}
\mathcal{H}^l (t)= \Big\{N \in \mathcal{N}^l: &N_{m,s}\geq N^l_{m,s+1}(t-1), \nonumber \\&\forall s\in \{1, \dots S-1\},\forall m \Big\}.
\end{align*}

Let us define a function $ F: \mathbb{Z}^{MS} \rightarrow \mathbb{R} $, referred to as the bias function, as follows. 
\begin{align}\label{DCN-eq:F}
F(q) = \max\Bigg\{1, \Big(\sum_m\sum_ssq_{m,s}\Big)^{\alpha}\Bigg\}, 
\end{align}
where $\alpha\in(0,1)$. At every slot $t$, for every server $l$, Q-BMW evaluates $\bar{N}^l(t)$ as in \eqref{eq:Q-BMW} and \eqref{eq:Q-BMW-2} where $\tau^l(t)$ is the last time a job migration happened at server $l$. It then prescribes a service configuration $ N^l(t) $ obtained from $ \bar{N}^l(t) $ following steps $ 2, 3 $ and $ 4 $ of Algorithm~\ref{Algo:Size}. Observe that Q-BMW offers a bias 
\[
\frac{\sum_m \Big(\sum_ssQ_{m,s}(t)\Big) \Big(\sum_s N_{m,s}\Big)-VC_1^l(N)}{F(Q(\tau^l(t)))}
\]
to the job service configurations that do not cause any job migration at server $ l $ at slot $ t $. Further, the bias at slot t also depends on the queue lengths at $\tau^l(t)$. The numerator of the bias term being linear and the denominator being sublinear, this term tends to assume high values for large queue lengths. Consequently, Q-BMW discourages job service configurations that cause job migrations when queue lengths are high. This is also reflected in Lemma~\ref{lemma:Tk} where the interval between such service configurations is lower bounded by an increasing function of queue lengths. We formally present Q-BMW as Algorithm~\ref{Algo:Q-BMW} below.
\begin{algorithm}
	\caption{Q-BMW}
	\label{Algo:Q-BMW}
	\begin{algorithmic}[1]
		\State Initialize $ \mathcal{H}^l(1) = \emptyset $, $ F(Q(\tau^l(1)))=1, \forall l \in[L]$, 
		\State For all $ t\geq 1 $, for all $l\in [L]$
		
		\begin{align} \label{eq:Q-BMW}
		&\mathcal{\bar{N}}^{l}(t) \nonumber\\
		 &=\argmax_{N \in \mathcal{N}^{l}}\Bigg(1+\frac{\mathbbm{1}_{\{N \in \mathcal{H}_l(t)\}}}{F\big(Q(\tau^l(t))\big)}\Bigg) \Bigg[\sum_m \Big(\sum_ssQ_{m,s}(t)\Big) \nonumber\\
		&\Big(\sum_s N_{m,s}\Big)- VC_1^l(N)\Bigg],
		\end{align}
		\begin{align}\label{eq:Q-BMW-2}
		N^l(t) \in \argmax_{N \in \bar{\mathcal{N}}^{l}(t)}\sum_{m,s}Q_{m,s}(t)N_{m,s}
		\end{align}
		\If {$\bar{N}^l(t)\in \mathcal{H}_l(t)$}
		\begin{align*}
		\tau^l(t+1) = \tau_l(t)
		\end{align*}
		\Else \begin{align*}
		\tau^l(t+1)=t
		\end{align*}
		\EndIf	\\

		\textbf{Obtaining $ N(t) $ from $ \bar{N}(t)$:}
		$N(t)$ is obtained from $\bar{N}(t)$ following steps $2,3$ and $4$ of Algorithm~\ref{Algo:Size}.
	\end{algorithmic}
\end{algorithm}
\paragraph*{Discussion}
Even though Q-BMW is motivated by the algorithms from ~\cite{Ping2017-DelayOpt} and \cite{chang2018DCNC} there are significant differences between them as explained below.  
\begin{itemize}
	\item In \cite{Ping2017-DelayOpt} and \cite{chang2018DCNC}, at time slot $t$ using the same user configuration and flow configuration respectively as at slot $t-1$ do not cause migration. Further, selecting any configuration other than the configuration at $t-1$ leads to migration. Whereas in Q-BMW at every slot $t$, there is a set of job service configurations,~$\mathcal{H}^l(t)$ that do not cause migration and these are different from the scheduled job service configuration at $ t-1$.
	\item In \cite{Ping2017-DelayOpt} and \cite{chang2018DCNC}, the authors assume that the server is completely stalled during the migration process. Whereas we assume that only execution of the migrated jobs is stalled and other jobs are executed without any interruption.
	\item In \cite{chang2018DCNC} the authors consider persistent flow based model and in our work, we consider a non-persistent, job based model. 
\end{itemize}
\par Let $ t^l_k, k \geq 1 $ be the successive time slots at which job migrations happen at server $ l; t^l_1 = 1 $ and for $ k \geq 1, $
\begin{align*}
t^l_{k+1} = \min\{t>t_k^l: \tau^l(t+1)=t\}.
\end{align*}
 Let us also define $ T^l_k, k \geq 1 $, to be the time intervals between $ k $th and $ (k+1) $th job migration epochs at server $ l $; $ T^l_k = t^l_{k+1} - t^l_k $. The following lemma gives a lower bound on $ T^l_k $.
\begin{lemma}\label{lemma:Tk}
 	Q-BMW, for all $k, $ and $l\in [L]$,
	\begin{align}
	T^l_k \geq \max \Big\{1, k_0 \Big(\sum_{m,s}sQ_{m,s}(t^l_k)-2V\Big)^{1-\alpha}\Big\}
	\end{align}
	where $k_0 =\frac{1}{M\Big(N^2_{\max}+MS^2A_{\max}N_{\max}\Big)} $
\end{lemma}
\begin{IEEEproof}
	The proof is similar to the proofs of 
	\cite[Lemma 3]{Ping2017-DelayOpt} and \cite[Lemma 1]{chang2018DCNC}.
\end{IEEEproof}
\par The following theorem establishes throughput optimality and cost optimality of Q-BMW.
\begin{theorem}\label{Theorem:Mig-Optimality}
 Q-BMW is throughput optimal and achieves $[O(1/V),O(V)]$ average cost delay tradeoff.	
\begin{align*}
&\lim_{T \to \infty}\frac{1}{T}\sum_{t=1}^T\sum_{m,s}\E[sQ_{m,s}(t)] \leq \frac{B+V\bar{C}_{\text{opt}}(\lambda+\e)}{\e}\\
&\lim_{T \to \infty}\frac{1}{T}\sum_{t=1}^T\sum_l\E[C_1^l(t)] \leq \frac{B}{V}+\bar{C}_{\text{opt}}(\lambda+\e)
\end{align*}
where $B$ depends on system parameters $\alpha, A_{\max}, N_{\max}, M,L$. and $\lambda+\e \in \Lambda$ for some $\e>0$.
\end{theorem}
\begin{IEEEproof}
	The proof is similar to the proofs of \cite[Theorem $1$]{Ping2017-DelayOpt} and \cite[Theorem $2$]{chang2018DCNC}.
\end{IEEEproof}
\begin{remark}\label{remark:vary-alpha}
	For a fixed $V$, the values of $\alpha$ close to zero imply more bias to the job service configurations that do not cause job migrations and hence fewer job migrations i.e., less workload is added to the system due to the migrations. The values of $\alpha$ close to one hence less bias to the job service configurations that do not cause job migrations and hence more job migrations~(see~Figure \ref{fig:alpha-mig-delay-switchins}(c)).  
\end{remark}

\subsection{A Modified Algorithm}
In this section, we propose a refined version of Q-BMW as Algorithm~\ref{Algo:Mig-Delay}. We first explain the motivation behind Algorithm~\ref{Algo:Mig-Delay}. Recall that Q-BMW offers a bias to the job service configurations that do not cause job migration but does not distinguish other configurations on the basis of the number of job migrations they cause. We modify this algorithm whereby each job service configuration receives a negative bias (or, penalty) commensurate with the number of job migrations is it causes (see~\eqref{DCN-eq:refined-Q-BMW}). In other words, a job service configuration that would lead to more job migrations receives more penalty than the one that would cause fewer job migrations, and the configuration that does not cause any job migration is not penalized. The penalty also increases sublinearly with the queue lengths implying a higher penalty for larger queue lengths as in Q-BMW. Consequently, Algorithm~\ref{Algo:Mig-Delay} also discourages job migrations when queue lengths are high. We have not been able to prove the throughput optimality of Algorithm~\ref{Algo:Mig-Delay}. 
However, we present its performance in Section~\ref{DCN-sec:simulations}. From Figures~\ref{fig:mig-delay-switchins} and~\ref{fig:alpha-mig-delay-switchins} we observe that Algorithm~\ref{Algo:Mig-Delay} has less number of mean job migrations and less number of active servers compared to Q-BMW.
\begin{algorithm}
	\caption{ Refined Q-BMW}
	\label{Algo:Mig-Delay}
	\begin{algorithmic}[1]
		\State At every slot $t$  given $Q(t)=q, N(t-1)=n$, for all $l \in [L]$,
		\begin{align} \label{DCN-eq:refined-Q-BMW}
		&\mathcal{\bar{N}}^{l}(t) \nonumber\\&= \argmax_{N \in \mathcal{N}^{l}} \Bigg[\sum_m \Big(\sum_ssq_{m,s}\Big) \Big(\sum_s N_{m,s}\Big)-VC_1^l(N)\nonumber\\
		  &-\Bigg(\sum_{m} \Big(\sum_{s>1}(n_{m,s}^l-N_{m,s-1})^+\Big)\Big(\sum_ssq_{m,s}\Big)\Bigg)^{1-\alpha}\Bigg]	
		\end{align}
		\begin{align}\label{DCN-eq:refined-Q-BMW-2}
		N^l(t) \in \argmax_{N \in \bar{\mathcal{N}}^{l}(t)}\sum_{m,s}Q_{m,s}(t)N_{m,s}
		\end{align}
		\textbf{Obtaining $ N(t) $ from $ \bar{N}(t)$:}
		$N(t)$ is obtained from $\bar{N}(t)$ following steps $2,3$ and $4$ of Algorithm~\ref{Algo:Size}.
	\end{algorithmic}
\end{algorithm}

\section{Simulations} \label{DCN-sec:simulations}
In this section, we demonstrate the performance of the scheduling algorithms proposed in this work. First, we demonstrate the performance of online scheduling algorithms (Algorithm~\ref{Algo:Size} and Algorithm~\ref{Algo:UK-algo-1}), later offline scheduling algorithms (Q-BMW and Algorithm~\ref{Algo:Mig-Delay}) via simulation. We also compare Algorithm~\ref{Algo:Size} and Algorithm~\ref{Algo:UK-algo-1} performance to that of non-preemptive and preemptive algorithms proposed in~\cite{stheja2012loadbal}. We refer to the latter algorithms as Non-Preemptive and Preemptive, respectively. 
\par We consider a simple set-up with ten identical servers $ (L=10) $. Each server can host three types of VMs $ (M=3) $. The resource constraints are such that $(0,0,2),(0,1,1)$ and $(1,1,0)$ are three maximal VM configurations for each server. Further, the maximum size of the jobs $10$ time-slots $(S=10)$. For any given job type, the job arrival rates are same for all sizes, which makes the average job size $5.5$ slots. We consider job arrival rate vector $ \lambda_{m,s} = L\rho \frac{m}{165} $, where $\rho \in \{0.2, 0.4, 0.6, 0.8, 0.9, 0.95, 0.97, 0.99, 1.01\}$. Here $\rho$ is load intensity which indicates how far the job arrival rate vector is from boundary of the capacity region. Further $\rho = 1$ corresponds to $ \lambda $ being on the Pareto boundary of the capacity region.
\par In Figures~\ref{fig:Affine_Throughput}, ~\ref{fig:Affine_src_V_U10}, and \ref{fig:Affine_jmc_U_V5} we demonstrate the performance of Algorithm~\ref{Algo:Size} with affine server running cost structure \eqref{DCN-eq:generic_src}. Towards this, we consider $(c_0, c_1, c_2, c_3=(1,2,6,3)$. 
In Figure~\ref{fig:Affine_Throughput} we demonstrate the throughput optimality of Algorithm~\ref{Algo:Size}. We chose weighing parameters $(V,U)=(5, 10)$. We observe that Algorithm~\ref{Algo:Size} is throughput optimal.
\begin{figure}[h!]
		\begin{center}
			\includegraphics[width=2.35in,height=1.8in]{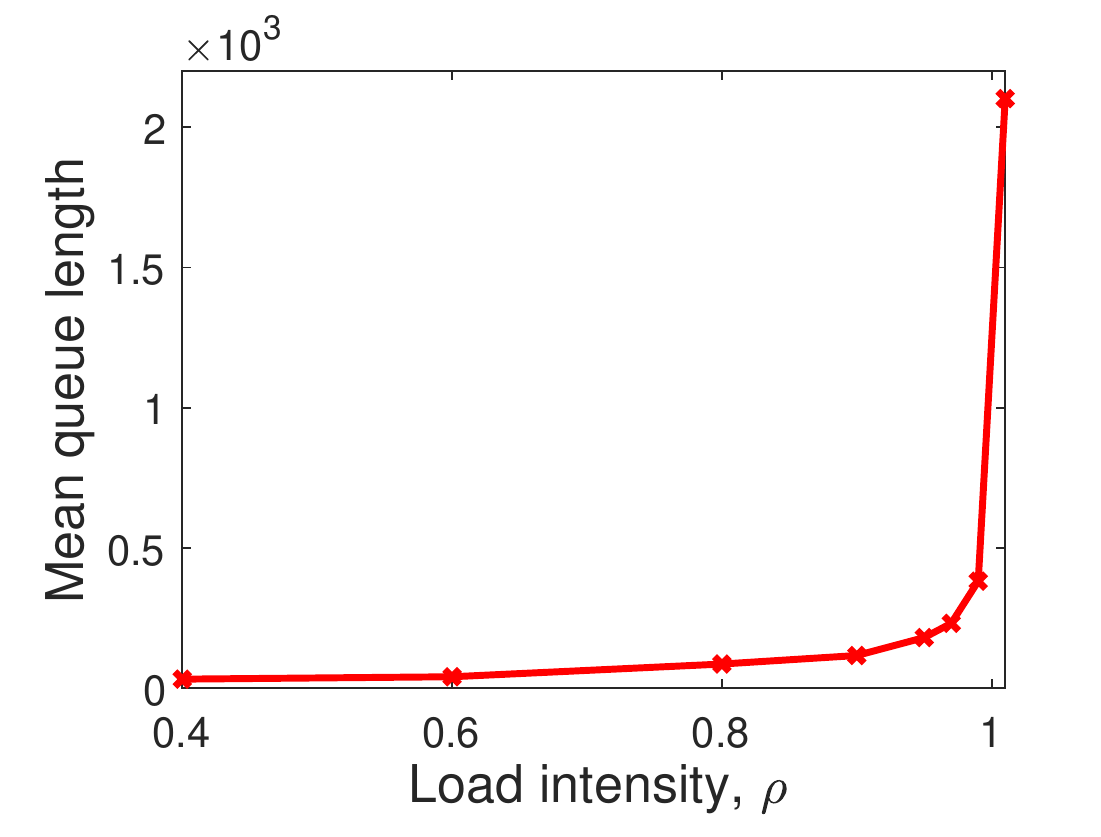}
			\caption{Illustrating the throughput optimality of Algorithm~\ref{Algo:Size} for $(V,U)=(5, 10)$ and $(c_0, c_1, c_2, c_3)=(1,2,6,3)$.}
			\label{fig:Affine_Throughput}
		\end{center}
\end{figure}
\par In Figure~\ref{fig:Affine_src_V_U10} we illustrate the impact of varying the weighing parameter $V$ on mean server running cost, mean queue lengths, and mean job migrations while keeping $U$ fixed. Further, we set $U=10, \rho=0.8$. The optimum server running cost for parameters $(c_0, c_1, c_2, c_3)$ and $\rho = 0.8$ is $ \bar{C}_{\text{opt}}(\lambda) = 56$. In Figure~\ref{fig:Affine_src_V_U10}(a) we observe that the mean server running cost 
decreases with increase in $V$ and approaches the optimum server running cost.
In Figure~\ref{fig:Affine_src_V_U10}(b) we observe that mean queue length increases with $V$. The reason is, for larger values of $V$, the algorithm prescribes a non-zero job configuration only if queue lengths are sufficiently large else the server is switched off. The jobs in the queues have to wait till a sufficient number of jobs arrive so that the algorithm can prescribe a non-zero job configuration. So the mean queue length increases with the increase in $V$. In Figure~\ref{fig:Affine_src_V_U10}(c) we observe that the mean job migrations increase with the increase in $V$. This reason is, for a larger value of $V$, there should be a sufficiently large number of jobs in the queues to schedule. If there are not enough jobs in the system then the scheduled jobs are preempted and cause an increase in the mean number of job migrations.
\begin{figure}[h!]
	\begin{center}
		\includegraphics[width=3.5in,height=1.3in]{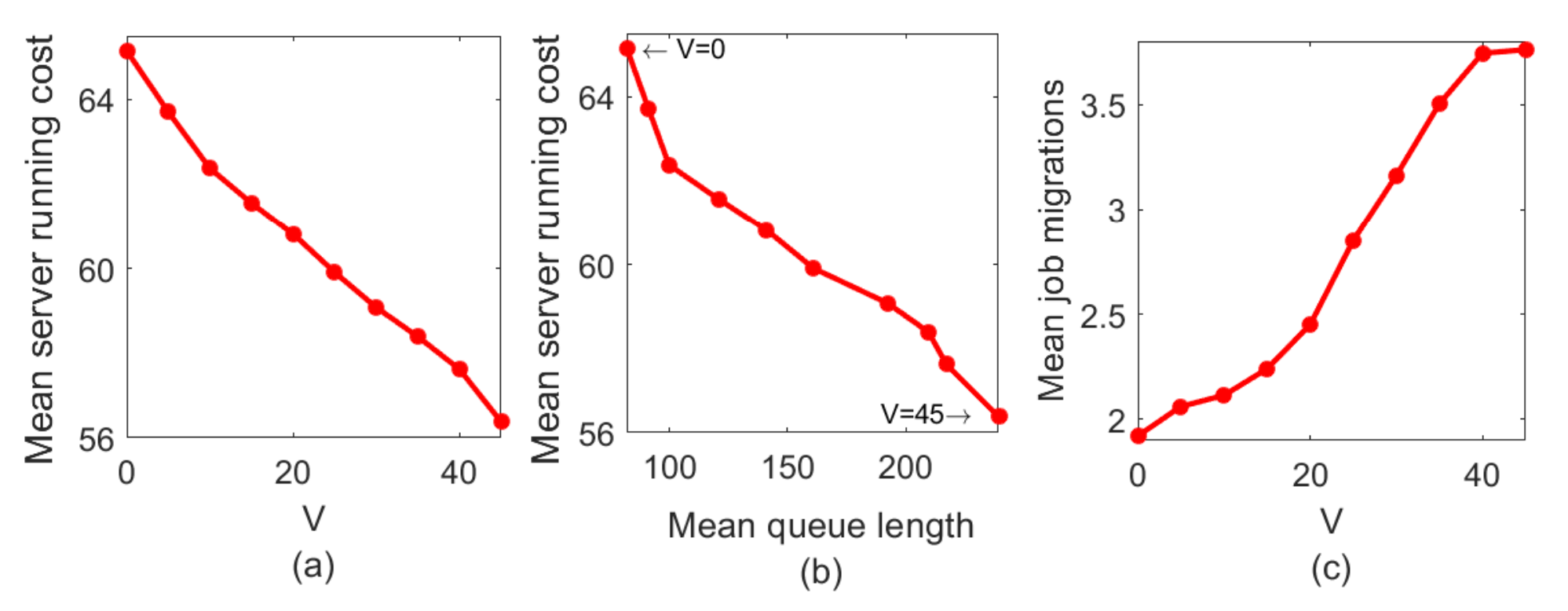}
		\caption{Illustrating the performance of Algorithm~\ref{Algo:Size} by varying $V$ and $U=10$, $(c_0, c_1, c_2, c_3)=(1,2,6,3)$.}
		\label{fig:Affine_src_V_U10}
	\end{center}
\end{figure}
\par In Figure~\ref{fig:Affine_jmc_U_V5} we illustrate the impact of varying the weighing parameter $U$ on mean job migrations, mean queue lengths, and mean server running cost while keeping $V$ fixed. Further, we set $V=5, \rho=0.8$. In Figure~\ref{fig:Affine_jmc_U_V5}(a) we observe that mean job migrations decrease with the increase in $U$. In Figure~\ref{fig:Affine_jmc_U_V5}(b) we observe that the mean queue length decreases with the increase in $U$. The reason is, as $U$ increases the job preemptions decrease i.e., the jobs running on the servers are not preempted and servers will be in on state, this increases average server on duration i.e., the server running cost and servers can serve more jobs. In Figure~\ref{fig:Affine_jmc_U_V5}(c) we observe that mean server running cost increases with the increase in $U$ as the higher values of $U$ discourage turning off the servers as explained for Figure~\ref{fig:Affine_jmc_U_V5}(b).
\begin{figure}[h!]
	\begin{center}
		\includegraphics[width=3.5in,height=1.3in]{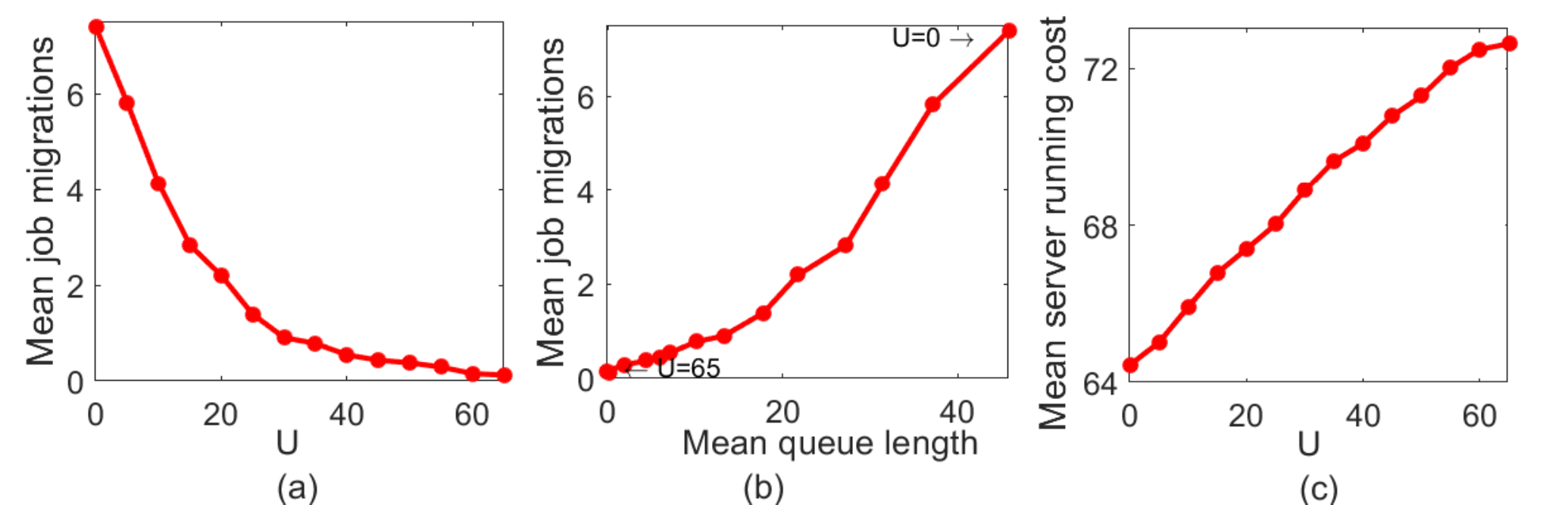}
		\caption{Illustrating the performance of Algorithm~\ref{Algo:Size} by varying $U$ and $V=10$, $(c_0, c_1, c_2, c_3)=(1,2,6,3)$.}
		\label{fig:Affine_jmc_U_V5}
	\end{center}
\end{figure}
\par For the rest of the simulations, we consider the binary server power model with $c_0 =1$~(see Remark~\ref{rem:bin-server}). Further, for $c_0 =1$, the mean server running is nothing but the mean number of active servers. In Figure \ref{fig:Throughput} we demonstrate the throughput optimality of the proposed algorithms. We chose $(V,U)$ as $(30, 10)$ and $(3,2)$ for Algorithm~\ref{Algo:Size} and Algorithm~\ref{Algo:UK-algo-1}, respectively. The selection of smaller values of $(V,U)$ for Algorithm~\ref{Algo:UK-algo-1} is explained in the following paragraphs. In Non-Preemptive we chose super time slot as $60$ time slots (see \cite{stheja2012loadbal} for more details). We observe that both the proposed algorithms are throughput optimal. Further, Preemptive has the least mean queue length among all algorithms since it does not consider any cost. In Preemptive, the servers are turned on whenever the job queues are nonempty, irrespective of the number of jobs available in the queues. Thus, preemptive has the least mean queue length among all the algorithms. Also, for higher load intensities, Non-Preemptive has the largest mean queue length.
\begin{figure}[h!]
		\begin{center}
			\includegraphics[width=2.35in,height=1.8in]{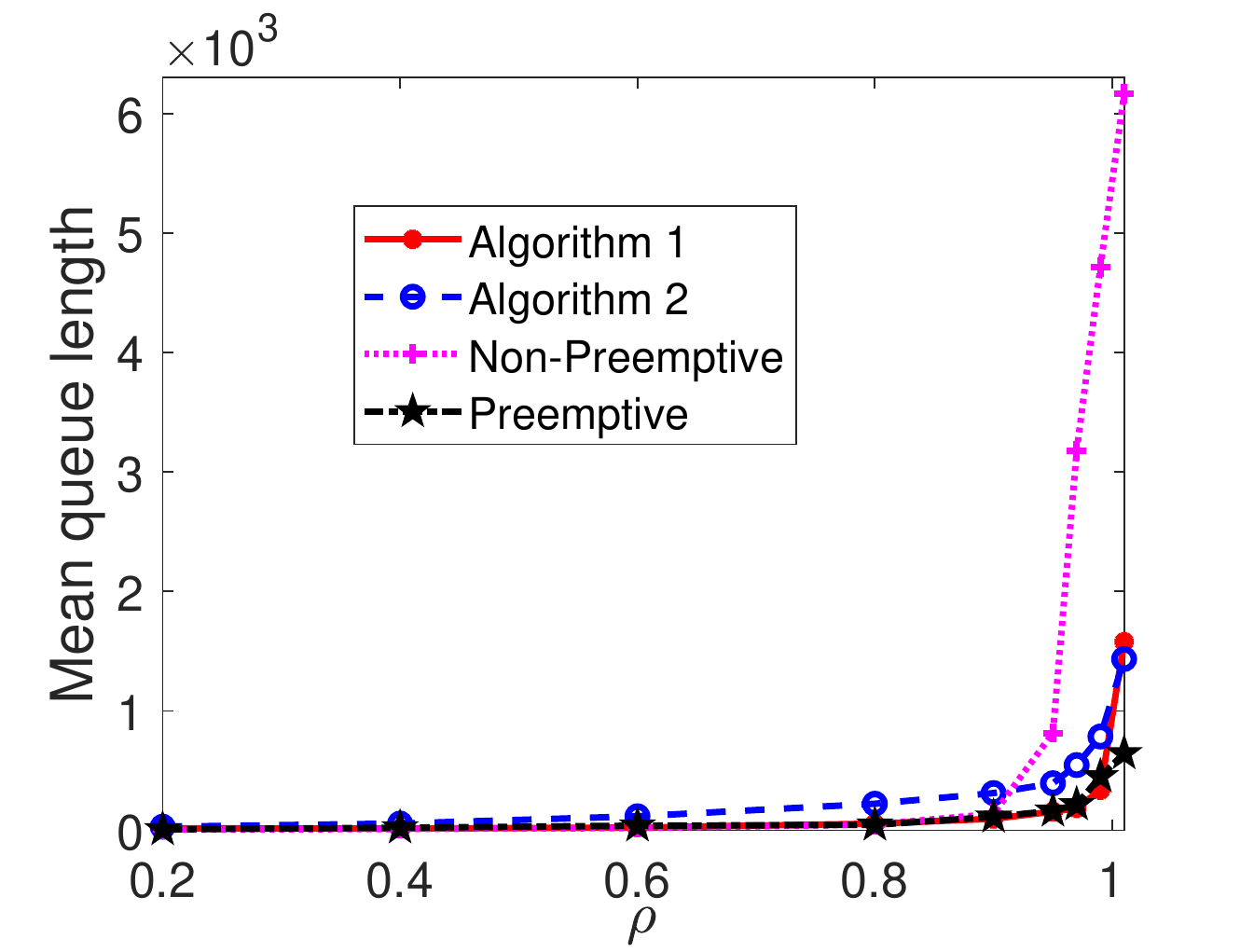}
			\caption{Illustrating the throughput optimality of different algorithms.}
			\label{fig:Throughput}
		\end{center}
\end{figure}

	\par In Figures~\ref{fig:Varying-V} and \ref{fig:QvsV} we illustrate the impact of varying the weighing parameter $V$ on the mean server running cost and mean queue lengths while keeping $U$ fixed. For Algorithm~\ref{Algo:Size}, we set $U=10$ and for Algorithm~\ref{Algo:UK-algo-1} we set $U=1$. We set $\rho = 0.8$ for further simulations.
	\begin{figure}[h!]
		\begin{center}
			\includegraphics[width=3.5in,height=1.7in]{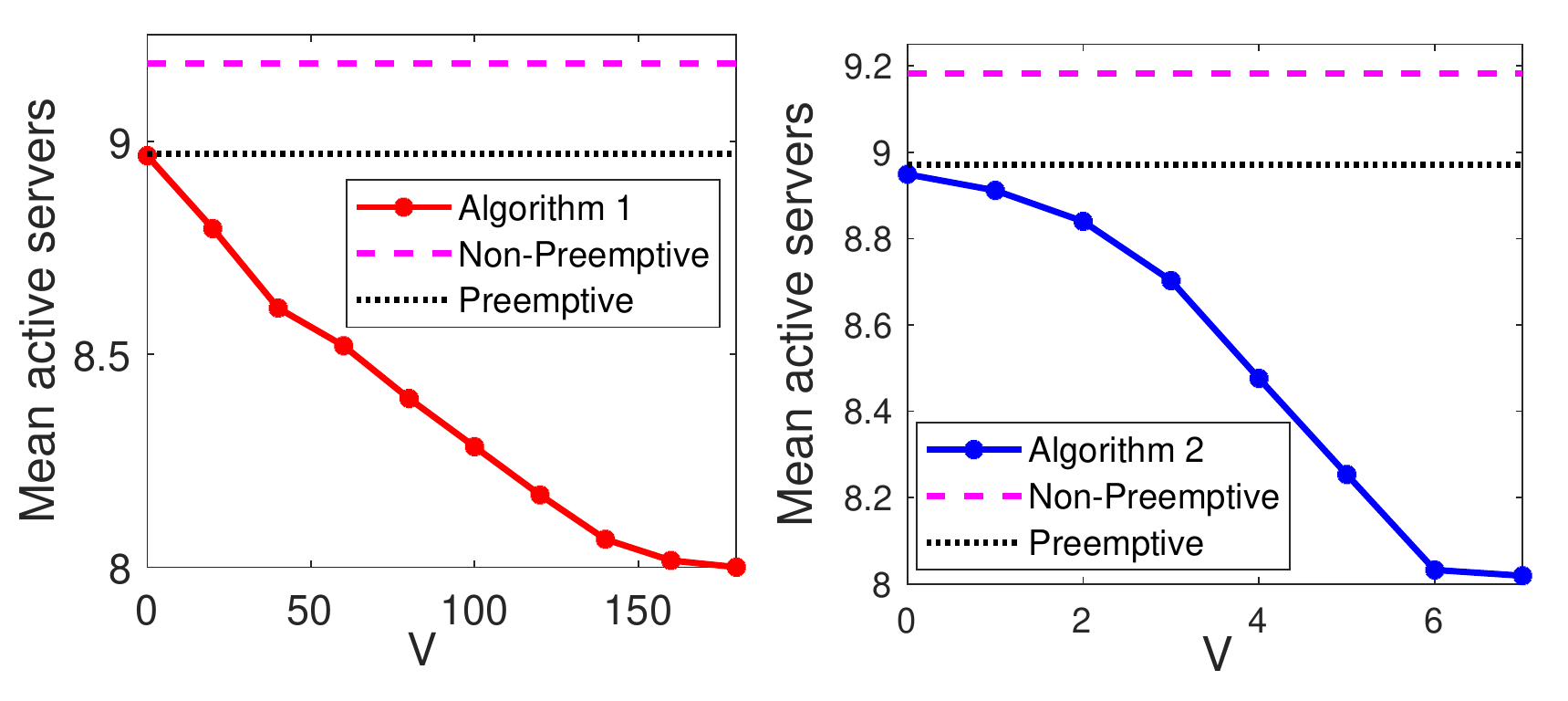}
			\caption{ Mean number of active servers Vs V, U=10 for Algorithm~\ref{Algo:Size} and U=1 for Algorithm~\ref{Algo:UK-algo-1}.}
			\label{fig:Varying-V}
		\end{center}
	\end{figure}
	\par In Figure~\ref{fig:Varying-V}, as $V$ increases the mean server running cost decreases for Algorithm~\ref{Algo:Size} and Algorithm~\ref{Algo:UK-algo-1}. We observe that Algorithm~\ref{Algo:UK-algo-1} achieves optimum server running cost for the lower value of $V$ compared to Algorithm~\ref{Algo:Size}. Since Algorithm~\ref{Algo:UK-algo-1} uses a weight function $g(x)=\log(1+x)$ to compute job service configuration, and more jobs accumulate at queues for smaller value of $ V $. This leads to complete utilization of servers and it is able to achieve optimum mean server running cost for smaller values of $V$. Preemptive is a work conserving policy, i.e., if queues are non-empty it switches on the servers irrespective of the number of jobs present. This leads to underutilization of servers and more mean server running cost. Further Non-Preemptive also under-utilizes servers for significant portions of time, leading to  a higher mean server running cost. 
		\begin{figure}[h!]
		\begin{center}
			\includegraphics[width=2.5in,height=1.9in]{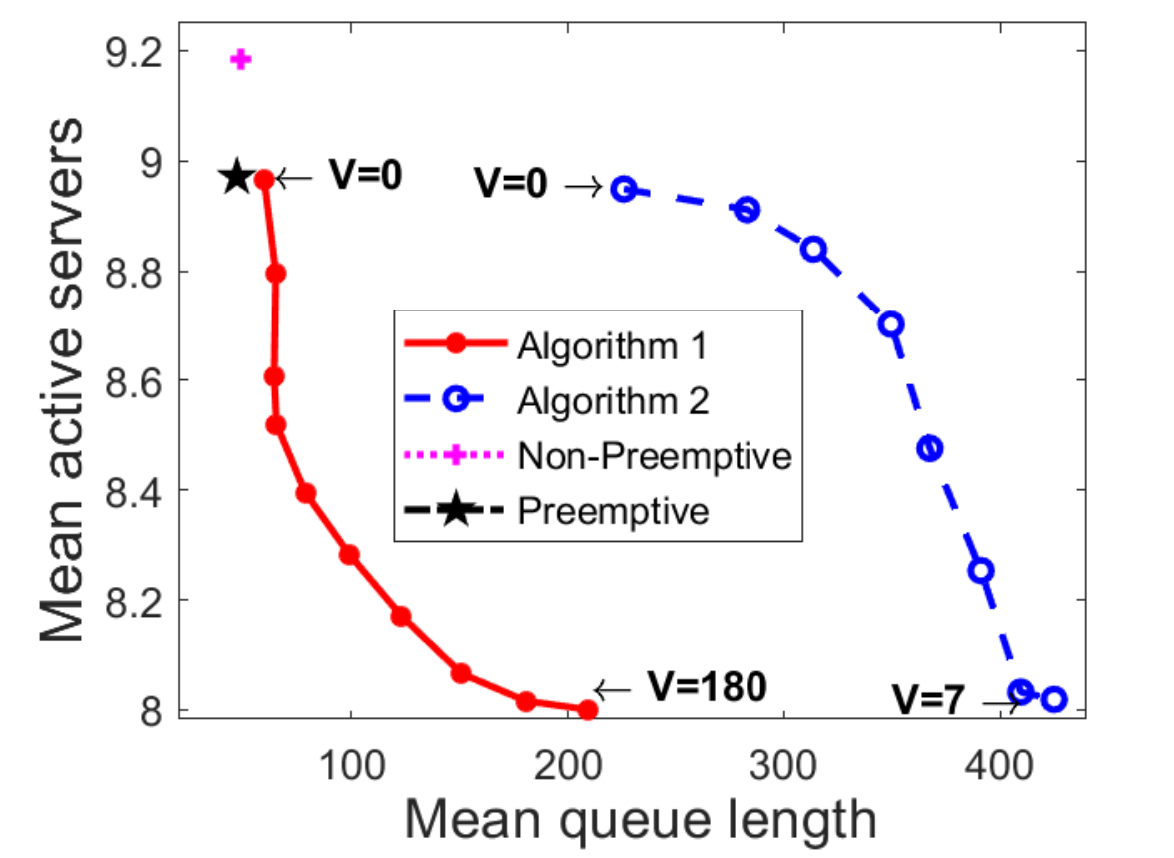}
			\caption{Mean Server Running Cost Vs Mean Queue Length, U=10 for Algorithm~\ref{Algo:Size} and U=1 for Algorithm~\ref{Algo:UK-algo-1}.}
			\label{fig:QvsV}
		\end{center}
	\end{figure}
	
\par In Figure~\ref{fig:QvsV}, as $ V $ increases the mean queue lengths increase whereas mean server running costs decrease. Expectedly, Algorithm \ref{Algo:Size} achieves a better queue length-cost tradeoff than Algorithm \ref{Algo:UK-algo-1}. Smaller values of $ V $ Algorithm~\ref{Algo:Size} performs almost identically to Preemptive, and both substantially outperform Nonpreemptive cost wise. Algorithm~\ref{Algo:UK-algo-1} has inferior mean queue length performance than Preemptive and Non-Preemptive but may lead to substantial savings on the mean server running cost. The latter depends on the choice of $ V $.	

\par In Figures~\ref{fig:Varying-U} and \ref{fig:QvsU}, we illustrate the impact of varying the weighing parameter $U$ on mean migration cost and mean queue lengths while keeping $V$ fixed. We set $V = 20$ and $V=6$ for Algorithm~\ref{Algo:Size} and Algorithm~\ref{Algo:UK-algo-1} respectively. 

\begin{figure}[h!]
	\begin{center}
        \includegraphics[width=3.65in,height=1.65in]{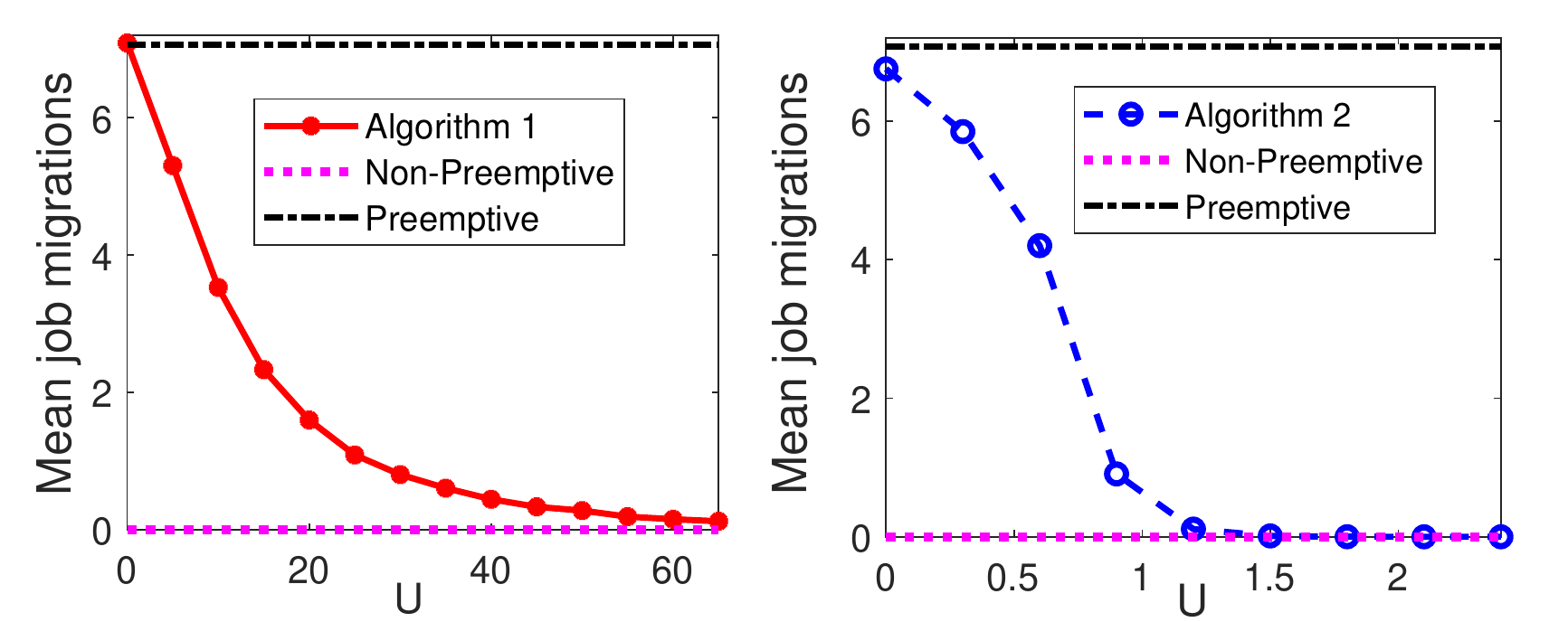}
		\caption{ Mean Job Migration Cost Vs U, V=20 for Algorithm~\ref{Algo:Size} and V=6 for Algorithm~\ref{Algo:UK-algo-1}.}
		\label{fig:Varying-U}
	\end{center}
\end{figure}
\par In Figure~\ref{fig:Varying-U} we observe that the mean migration cost decreases with the increase in $U$ for Algorithm~\ref{Algo:Size} and Algorithm~\ref{Algo:UK-algo-1}. We observe that Algorithm~\ref{Algo:UK-algo-1} achieves zero mean migration cost for smaller values of $ U $ than Algorithm \ref{Algo:Size}. Non-Preemptive computes VM configuration at every time slot $t$ such that no job migration takes place. This leads to zero migration cost. Preemptive computes job schedule at every time slot $ t $ independent of the previous job schedule. This leads to inferior mean migration cost performance compared to all other algorithms.
\begin{figure}[h!]
	\begin{center}
		\includegraphics[width=2.5in,height=2in]{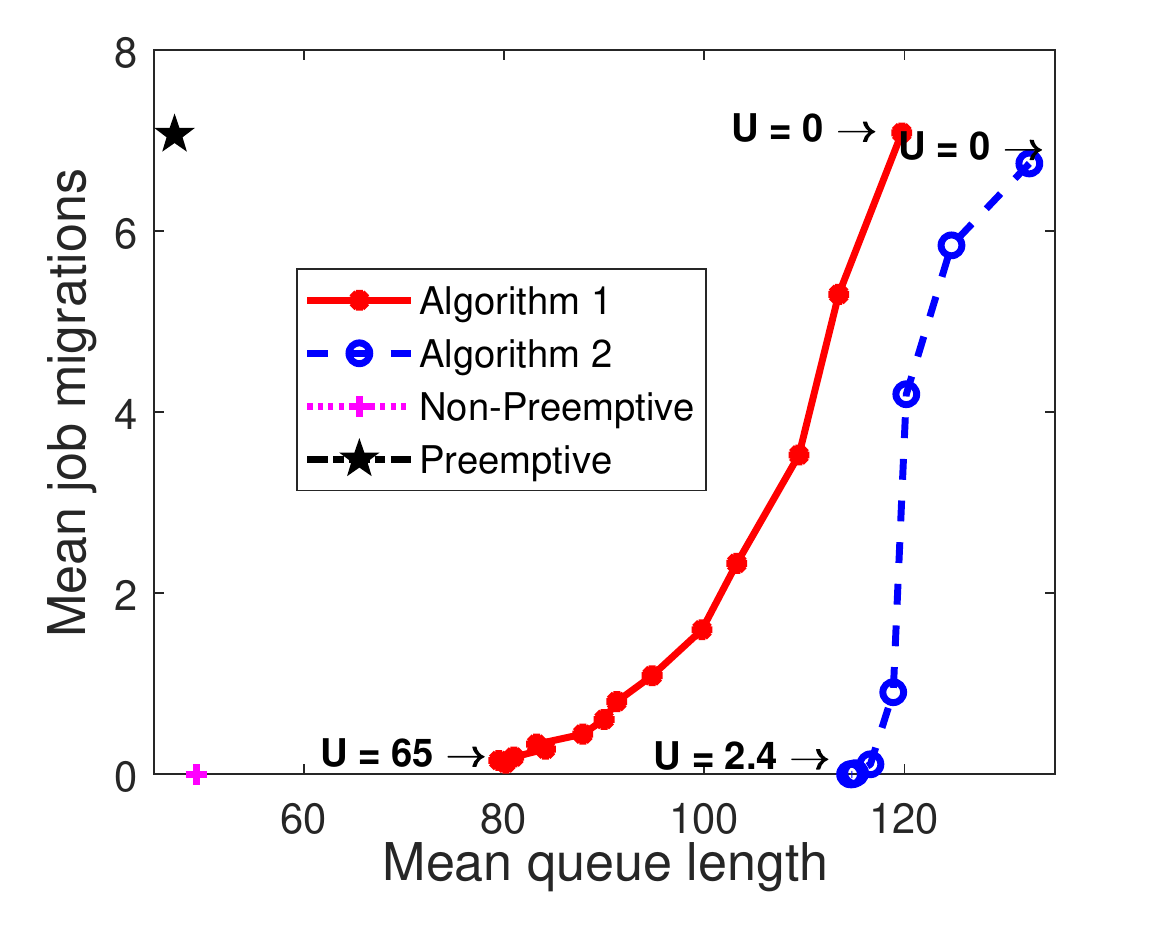}
		\caption{Mean Job Migration Cost Ons Vs Mean Queue Length,V=20 for Algorithm~\ref{Algo:Size} and V=6 for Algorithm~\ref{Algo:UK-algo-1}.}
		\label{fig:QvsU}
	\end{center}
\end{figure}
\par In Figure~\ref{fig:QvsU}, mean queue lengths decrease as $U$ increases. At larger values of $U$ job migrations are discouraged leading to switching on of more servers. This mitigates the effect of $V$. As more servers are available more jobs are served to lead to less mean queue length. Since Preemptive is work conserving policy and chooses the best VM configuration irrespective of the previous configuration has low mean queue length and more mean job migration cost. Both Algorithm~\ref{Algo:Size} and Algorithm~\ref{Algo:UK-algo-1} substantially outperform Preemptive cost wise.	

To study the performance of Q-BMW and Algorithm~\ref{Algo:Mig-Delay} we consider a simple setup with ten identical servers, $L=10$. Each server can host two types of VMs, $M=2$. The possible maximal VM configurations are  $(0,1),(3,0)$.
\par In Figure~\ref{fig:mig-delay-q-len} we demonstrate the throughput optimality of Q-BMW and Algorithm~\ref{Algo:Mig-Delay}. We observe that both are throughput optimal and mean queue length increase with $\alpha$. 
\begin{figure}[h!]
	\begin{center}
		\includegraphics[width=2.3in,height=1.85in]{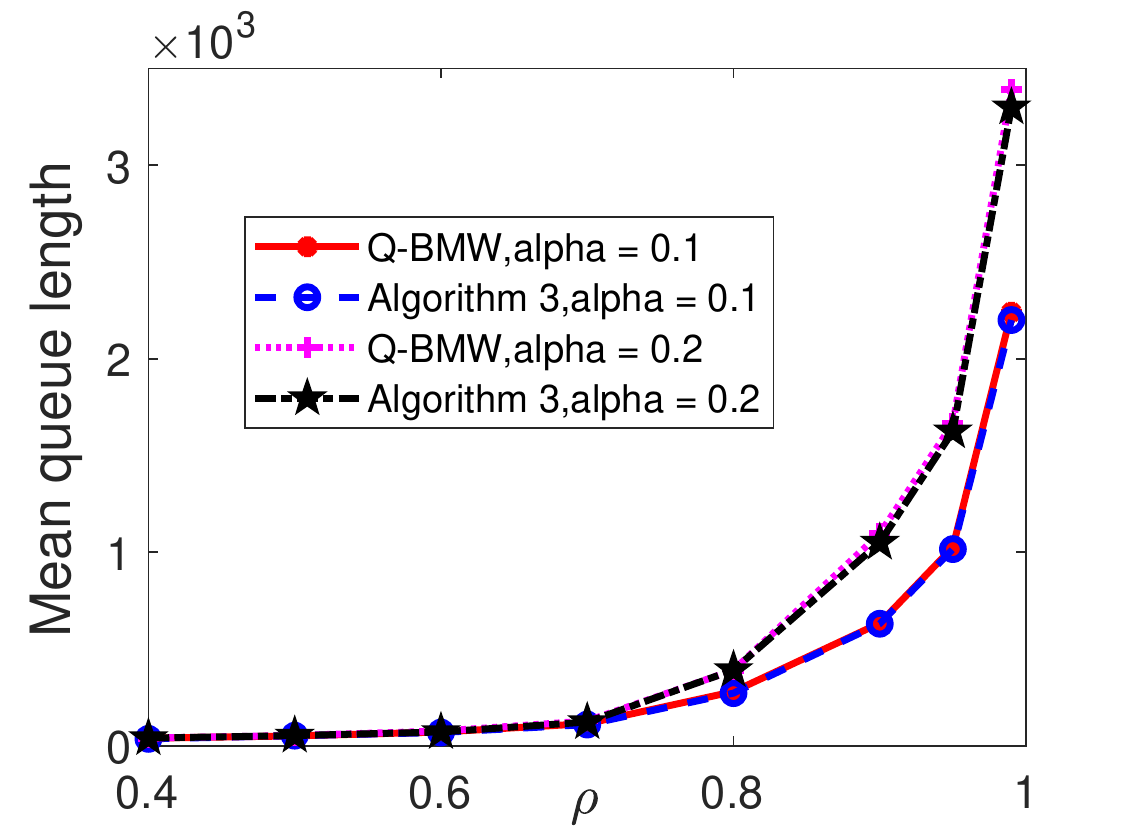}
		\caption{Illustrating throughput optimality of Q-BMW and Algorithm~\ref{Algo:Mig-Delay} for $\alpha=\{0.1, 0.2\}$, and $V=10$}
		\label{fig:mig-delay-q-len}
	\end{center}
\end{figure}
In Figure~\ref{fig:mig-delay-switchins} we illustrate the impact of the weighing parameter $V$ on the mean number of active servers, mean queue length, and mean job migrations for Q-BMW and Algorithm~\ref{Algo:Mig-Delay}. We consider $\rho = 0.8$ and vary the weighing parameter $V$. For $\rho = 0.8$ and $L=10$ the $ C_{opt}^1(\lambda) =8$. In Figure~\ref{fig:mig-delay-switchins}(a) we observe that as $V$ increases the mean  number of active servers approach the optimal value for both the algorithms. In Figure~\ref{fig:mig-delay-switchins}(b), (c) we observe that the mean queue length and mean job migrations increase with the increase in $V$ and the same explanation as for Figure~\ref{fig:Affine_src_V_U10} holds here.

\begin{figure}[h!]
	\begin{center}
		\includegraphics[width=3.5in,height=1.2in]{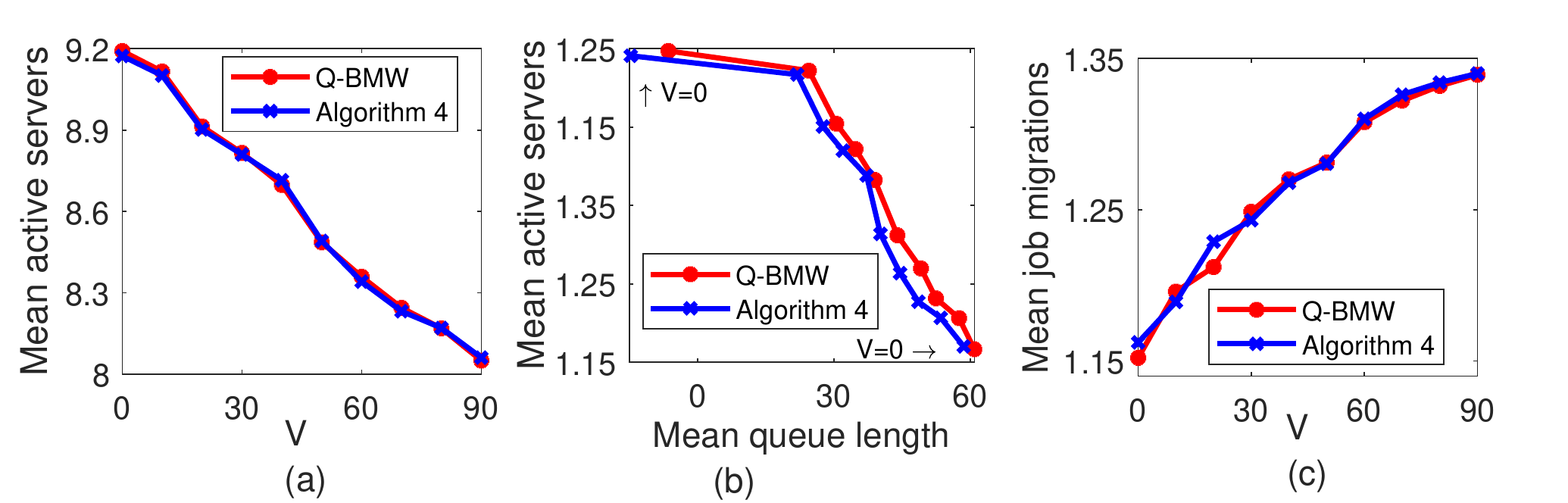}
		\caption{Illustrating the cost optimality of Q-BMW and Algorithm~\ref{Algo:Mig-Delay} for $\alpha = 0.1, \rho = 0.8$}
		\label{fig:mig-delay-switchins}
	\end{center}
\end{figure}
In Figure~\ref{fig:alpha-mig-delay-switchins} we illustrate the impact of  $\alpha$ on mean queue length, the mean number of active servers, and mean job migrations for Q-BMW and Algorithm~\ref{Algo:Mig-Delay}. We consider $\rho = 0.8$ and $V=10$. In Figure~\ref{fig:mig-delay-switchins}(a) we observe that the mean queue length increases with the increase in $\alpha$. The reason is the job migrations increase with $\alpha$ and thus we observe an increase in queue length~(see Remark~\ref{remark:vary-alpha}). In Figure~\ref{fig:mig-delay-switchins}(b) we observe that the mean number of active servers decrease with the increase $\alpha$ and the same explanation as for Figure~\ref{fig:Affine_jmc_U_V5} holds here.

\begin{figure}[h!]
	\begin{center}
		\includegraphics[width=3.5in,height=1.2in]{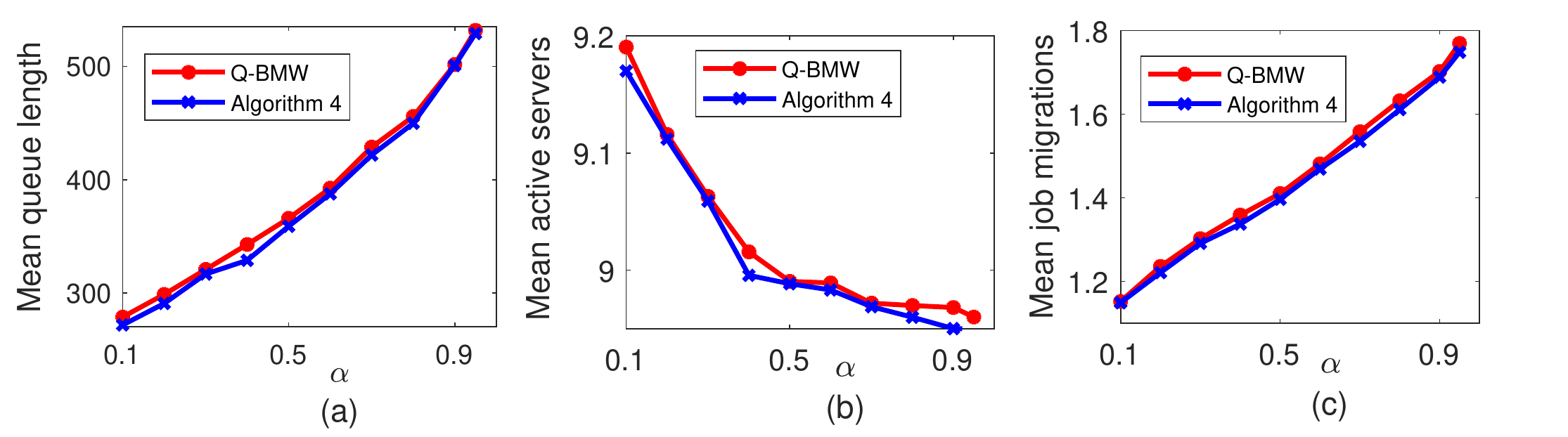}
		\caption{Illustrating the cost optimality of Q-BMW and Algorithm~\ref{Algo:Mig-Delay} for $V = 10, \rho = 0.8$}
		\label{fig:alpha-mig-delay-switchins}
	\end{center}
\end{figure}

\section{Conclusion} \label{DCN-sec:conclusion}
	We study the problem of job scheduling and migration in cloud computing clusters, with the objective of
	minimizing server running and job migration costs while rendering the job queues stable. Our approach is based on
	drift plus penalty framework. In the online migration model, We propose two job scheduling
	algorithms, the first one requiring job size information and the second one not requiring it, both incurring
	average costs arbitrarily close to the optimal cost. In offline migration, we present two throughput and cost optimal algorithms. 
	
\bibliography{DCN}		
	\appendices
	\section{Proof of lemma \ref{Lemma:Markov-The-Best}} \label{Proof:Markovian the Best}
	Let us consider an arbitrary scheduling policy,
	prescribing a sequence of VM configurations $N(t), t \geq 1$.\footnote{Proof of idea behind this theorem is similar to the Theorem 1 from \cite{Neely-Energy-Opt}. But we do not use Caratheodory theorem since the number of VM configurations possible and number of job migrations that can happen are finite in number.}
	Let us fix a $T > 0$.
	For all $l \in [L]$ and $N^l \in {\cal N}^l$, define
	$\tau^l(N^l;T) \subset [T]$ to be the set of time slots that see schedule $N^l$ as the previous configuration
	at server $l$.
	\[\tau^l(N^l;T) = \{t \in [T]: N^l(t-1) = N^l\}.\]
	Similarly, define $\P^l(N'^l,N^l;T) $ to be the fraction of slots with configuration $N$, conditioned on previous slot configuration being $N'^l$
	and $\pi^l(N^l;T)$ to be the~(unconditioned)fraction of slots with configuration $N$.
	\begin{align*}
	P^l(N'^l,N^l;T) &= \frac{1}{T}\sum_{t \in \tau^l(N'^l;T)} \mathbbm{1}_{\{N^l(t) = N^l\}},\\
	\pi^l(N^l;T) &= \frac{1}{T}\sum_{t \in [T]} \mathbbm{1}_{\{N^l(t) = N^l\}}.
	\end{align*}
	Observe that
	\[
	\pi^l(N^l;T) = \sum_{N'^l \in {\cal N}^l} \pi^l(N'^l;T)P^l(N'^l,N^l;T).
	\]
	Further, define empirical average service rate $\bar{\mu}^l(T)$, empirical average server running cost $\bar{C}_1^l(T)$, empirical average
	migration cost $\bar{C}_2^l(T)$ and empirical weighted average cost as below.
	\begin{align*}
	\bar{\mu}^l_{m,s}(T) &= \sum_{N^l}\pi^l(N^l;T)(N^l_{m,s}-N^l_{m,s+1}), \ \forall m,s,\\
	\bar{C}_1^l(T) =& \sum_{N^l}\pi^l(N^l;T) \mathbbm{1}_{\{\sum_{m,s}N^l_{m,s} > 0\}}, \\
	\bar{C}_2^l(T) =& \sum_{N^l}\pi^l(N^l;T)\sum_{N'^l}P^l(N'^l,N^l;T) \\&\sum_{m,s\geq2}(N'^l_{m,s}-N^l_{m,s-1})^+ \\
	\bar{C}^l(T) = &V\bar{C}_1^l(T) + U\bar{C}_2^l(T).
	\end{align*}
	Since we have finitely many sequences, we can find a subsequence of $T_r \in \mathbb{Z}_+$ such that
	the following limits hold.
	\begin{align*}
	\lim_{r \to \infty} \mu^l_{m,s}(T_r) &=  \mu^l_{m,s}, \forall l,m,s,\\
	\lim_{r \to \infty} \pi^l(N^l;T_r) &=  \pi^l(N^l), \forall l, N^l \in {\cal N}^l,\\
	\lim_{r \to \infty} P^l(N'^l,N^l;T_r) &=  P^l_{N'^l,N^l}, \forall l, N'^l,N^l \in {\cal N}^l,\\
	\lim_{r \to \infty} \bar{C}^l(T_r) &=  \bar{C}^l, \forall l.\\
	\end{align*}
	Clearly, for all $l$,
	\begin{align*}
	\bar{\mu}^l_{m,s} = &\sum_{N^l}\pi^l(N^l)(N^l_{m,s}-N^l_{m,s+1}), \ \forall m,s,\\
	\bar{C}^l = &U\sum_{N^l}\pi^l(N^l)\sum_{N'^l}P^l_{N'^l,N^l}\sum_{m,s\geq2}(N'^l_{m,s}-N^l_{m,s-1})^+ \\
	&+ V\sum_{N^l}\pi^l(N^l) \mathbbm{1}_{\{\sum_{m,s}N^l_{m,s} > 0\}}.
	\end{align*}
	Notice that $P^l \in {\cal S}^l, \pi^l \in {\cal P}^l$
	and $\pi^l P^l =  \pi^l$. Also, stability implies that there exist $(\lambda^l)_l$
	such that $\lambda  = \sum_l \lambda^l$
	and
	\[
	\lambda_{m,s}^l \leq \sum_{N^l \in {\cal N}^l}\pi^l(N^l)(N^l_{m,s}-N^l_{m,s+1}).
	\]
	Since $C_{\text{opt}}(\lambda)$ is the optimal value of over all
	such $(\lambda^l)_l, P^l$ and $\pi^l$, $\bar{C} = \sum_l \bar{C}^l  \geq C_{\text{opt}}(\lambda)$.

	\section{Proof of Theorem \ref{Theorem:Size-Algo-Optimality}} \label{Proof:Size-Algo-Optimality}
We begin with introducing some notation. Let $W_m(t)$ be the workload of type-$m$ that arrives at time $t$~$(W_m(t) = \sum_ssA_{m,s}(t))$ and $J_m(t) = \sum_s sQ_{m,s}(t)$ be the workload backlog of type-$m$ jobs at time $t$. Further, let $N^l_m(t)$ be the number of type-$m$ VMs served at time $t$;$ N^l_m(t) = \sum_s N^l_m(t)$. The workload of type-$m$ evolves as
	\begin{align*}
	J_m(t+1) = J_m(t)-\sum_l N^l_m(t)+W_m(t).
	\end{align*}
Squaring on both sides of above equation and rearranging the
terms
	\begin{align*}
		&(J_m(t+1))^2-(J_m(t))^2 \\
		=&\Big(W_m(t)-\sum_l N^l_m(t)\Big)^2+2J_m(t)\Big(W_m(t)-\sum_l N^l_m(t)\Big)\\
	\leq &\Big(\sum_l N^l_m(t)\Big)^2+\Big(W_m(t)\Big)^2\\
	&+2J_m(t)(W_m(t)-\sum_l N^l_m(t))	\\
	\leq &  2B_1+2J_m(t)\Big(W_m(t)-\sum_l N^l_m(t)\Big)
	\end{align*}
	where $ B_1 = \frac{1}{2}\Big((LN_{\max})^2+S^2A^2_{\max}\Big) $. Let us define $L(t) = \frac{1}{2}\sum_m(J_m(t))^2$ and $\Delta(t)=L(t+1)-L(t)$. The above inequality can be written as
	\begin{align} \label{DCN-eq:delta(t)}
	\Delta(t) \leq MB_1+\sum_m J_m(t)W_m(t)-\sum_m\sum_l J_m(t)N^l_m(t).
	\end{align}
	Let us now consider last term in \eqref{DCN-eq:delta(t)},
	\begin{align} \label{DCN-eq:N-bar}
	\sum_{m,l}& J_m(t)N^l_m(t) \nonumber\\
	=&\sum_{m,l} J_m(t)N^l_m(t)\mathbbm{1}_{\{J_m(t)< LN_{\max}S\}} \nonumber \\
	  &+\sum_{m,l} J_m(t)N^l_m(t)\mathbbm{1}_{\{J_m(t)\geq LN_{\max}S\}} \nonumber \\	 
	\geq  &\sum_{m,l} J_m(t)(\bar{N}^l_m(t)-N_{\max})\mathbbm{1}_{\{J_m(t) < LN_{\max}S\}} \nonumber\\&+\sum_{m,l} J_m(t)\bar{N}^l_m(t)\mathbbm{1}_{\{J_m(t)\geq LN_{\max}S\}} \nonumber\\
	\geq  & \sum_{m,l} J_m(t)\bar{N}^l_m(t)-MLN^2_{\max}S
	\end{align}
	The second last inequality follows since, whenever $J_m(t) \geq LN_{\max}S$, Step-$1$ in Algorithm ensures that $\sum_s\bar{N}^l_{m,s}=\sum_sN^l_{m,s}$ i.e., $\bar{N}^l_m(t)=N^l_m(t)$.
	Using \eqref{DCN-eq:N-bar} in \eqref{DCN-eq:delta(t)},
	\begin{align*}
	\Delta(t) \leq B_2 + \sum_m J_m(t)W_m(t)-\sum_{m,l} J_m(t)\bar{N}^l_m(t)
	\end{align*}
	where $B_2 = MB_1+LMSN^2_{\max}$.
	We now add server running and job migration costs on both the sides of the above inequality and take conditional expectation, conditioned on $Y(t) := (J(t), N(t-1))$.
	\begin{align}\label{DCN-eq:drift+cost_1}
	\E&\Bigg[\Delta(t)+\sum_l(VC^l_1(\bar{N}^l(t))+UC^l_2(N^l(t-1),\bar{N}^l(t)))|Y(t)\Bigg] \nonumber\\
	\leq& B_2+\sum_{m,s}J_m(t)s\lambda_{m,s}-\Bigg[\sum_{m,l} J_m(t)\bar{N}^l_m(t)\nonumber \\
	&-\sum_l(VC^l_1(\bar{N}^l(t))+UC^l_2(N^l(t-1),\bar{N}^l(t)))\Bigg].	
	\end{align}
	In order to bound the job migration cost of Algorithm~\ref{Algo:Size} we consider another algorithm that accounts for only the server running cost. In particular, it yields a configuration 
	\begin{align*}
		\eta^l(t) \in \argmax_{N \in \mathcal{N}^l}\sum_mJ_m(t)N_m-VC^l_1(N^l(t))
	\end{align*}
	for server $l$ at time $t$. From the definition of $\bar{N}^l_m(t)$ (see Step~$1$ in Algorithm~\ref{Algo:Size}),
	\begin{align}\label{DCN-eq:bar-star}
	\Bigg(\sum_m&J_m(t)\bar{N}^l_m(t)-VC^l_1(\bar{N}^l(t))\nonumber \\&-UC^l_2(N^l(t-1),\bar{N}^l(t))\Bigg)\geq \Bigg(\sum_mJ_m(t)\eta^l_m(t) \nonumber\\&-U\sum_{m,s=2}C_2^l(N^l,\eta^l(t))-VC_1^l(\eta^l(t))\Bigg).
	\end{align}
Using this inequality in \eqref{DCN-eq:drift+cost_1} and also noticing 
that $\sum_{m,s =2} N^l_{m,s} \leq N_{\max}$,
	\begin{align}\label{DCN-eq:drift+cost_2}
	\E&\Bigg[\Delta(t)+\sum_l(VC^l_1(\bar{N}^l(t))+UC^l_2(N^l(t-1),\bar{N}^l(t)))|Y(t)\Bigg] \nonumber\\
	\leq& B_2+ULN_{\max} +\sum_{m,s}J_m(t)s\lambda_{m,s}\nonumber \\
	&-\Bigg[\sum_{m,l} J_m(t)\eta^l_m(t)-V\sum_l C_1^l(\eta^l(t))  \Bigg].
	\end{align}
	Let an optimal solution of Problem~\eqref{DCN-eq:Opt-2}, for arrival rate vector $\lambda +\epsilon$, be attained at $\bar{\pi} \in \cal{P}$. Using the definition of $\eta^l(t)$ in \eqref{DCN-eq:drift+cost_2}, we can write
    \begin{align}\label{DCN-eq:drift+cost_3}
    \E&\Bigg[\Delta(t)+\sum_l(VC^l_1(\bar{N}^l(t))+UC^l_2(N^l(t-1),\bar{N}^l(t)))|Y(t)\Bigg] \nonumber\\
    \leq& B_2+ULN_{\max} +\sum_{m,s}J_m(t)s\lambda_{m,s}-\Bigg[\sum_{m,l} J_m(t)\sum_{N \in \mathcal{N}^l}\bar{\pi}^l_{N^l}N^l_m\nonumber \\&-V\sum_l\Big(\sum_{N \in \mathcal{N}^l}\bar{\pi}^l_{N^l}\sum_mc_mN_m-c_0(1-\pi_0^l)\Big)  \nonumber \Bigg] \nonumber \\
    \leq & B_2+ULN_{\max} +\sum_{m,s}J_m(t)s\lambda_{m,s} \nonumber \\ &-\sum_mJ_m(t)\Big(\sum_ss\lambda_{m,s}+\e\Big)+V\bar{C}_{\text{opt}}(\lambda+\e) \nonumber \\
    \leq & B_2+ULN_{\max} -\e\sum_mJ_m(t)+ V\bar{C}_{\text{opt}}(\lambda+\e)
    \end{align}
    where the second last inequality follows because $\bar{C}_{\text{opt}}(\lambda+\epsilon)$ is an upper bound on the server running cost incurred by $\bar{\pi}$. Taking expectation with respect to $N(t-1)$,
	\begin{align*}
	\E&\Bigg[\Delta(t)+\sum_l(VC^l_1(\bar{N}^l(t))+UC^l_2(N^l(t-1),\bar{N}^l(t)))|J(t)\Bigg] \nonumber\\
	\leq & B_2+ULN_{\max} -\e\sum_mJ_m(t)+ V\bar{C}_{\text{opt}}(\lambda+\e)
	\end{align*}
	Summing over $t=0 \dots T-1$ and taking expectation,
	\begin{align}\label{DCN-eq:drift+cost_5}
	\E&\Bigg[L(T)-L(0) \nonumber \\
	&+\sum_{t=0}^{T-1}
	\Big[\sum_l(VC^l_1(\bar{N}^l(t))+UC^l_2(N^l(t-1),\bar{N}^l(t)))\Big]\Bigg] \nonumber\\
	\leq & TB_2+TULN_{\max} -\e\sum_t\sum_m\E[J_m(t)]+ TV\bar{C}_{\text{opt}}(\lambda+\e).
	\end{align}	
	Following standard steps of drift plus penalty technique we obtain (see\cite{Drift-plus-Penalty}),
	\begin{align*}
	\lim_{T \to \infty} \frac{1}{T}\sum_t^T\sum_m\E[J_m(t)] \leq \frac{B_2+ULN_{\max}+V\bar{C}_{\text{opt}}(\lambda+\e)}{\e},
	\end{align*}
	\begin{align*}
	\lim_{T \to \infty} \frac{1}{T}\sum_t^T\sum_{m,l}\E[C^l_1(\bar{N}^l(t))]  
	\leq  \frac{B_2+ULN_{\max}}{V}+\bar{C}_{\text{opt}}(\lambda+\e).
	\end{align*}
Since $\epsilon > 0$ can be chosen to be arbitrarily small and $\bar{C}_{\text{opt}}(\cdot)$ is a continuous function~(see \cite[Lemma~1]{subhashini2017augment},) we can replace $\bar{C}_{\text{opt}}(\lambda+\epsilon)$ with $\bar{C}_{\text{opt}}(\lambda)$ in the above inequalities. We thus establish parts (1) and (2) of the theorem. 	
	
\par To prove part~$(c)$ of theorem, at every $t$, we consider a job service configuration $\bar{\eta}(t)$ that leads to zero number of migrations. More specifically,
	\begin{align*}
	\bar{\eta}(t)^l_{m,s}(t) &= \bar{N}^l_{m,s+1}(t-1), l \in [L], m\in [M], s \in [S-1], \\
	\bar{\eta}(t)^l_{m,S}(t) &= \bar{N}^l_{m,1}(t-1), l \in [L], m\in [M],\\
	\bar{\eta}(t)^l_m(t) &= \bar{N}_m(t-1).
	\end{align*} 
	Clearly $\bar{\eta}(t)$ do not cause any job migrations.
	From the definition of $\bar{N}(t)$ we can write,
	\begin{align*}
	\sum_m &J_m(t)\bar{\eta}^l(t)_m(t)-VC_1^l(\bar{\eta}^l(t))\leq \sum_mJ_m(t)\bar{N}^l_m(t) \\ 	-&VC_1^l(\bar{N}^l(t)) -U\sum_{m,s=2}\Big(N^l_{m,s}(t-1)-\bar{N}^l_{m,s-1}(t)\Big)^+.
	\end{align*}
	Rearranging the terms
	\begin{align*}
	U&\sum_{m,s=2}\Big(N^l_{m,s}(t-1)-\bar{N}^l_{m,s-1}(t)\Big)^+ \\
	\leq& V\big[C_1^l(\bar{\eta}^l(t))-C_1^l(\bar{N}^l(t))\big]+\sum_mJ_m(t)\bar{N}^l_m(t)\\
	&-\sum_mJ_m(t)\bar{\eta}(t)^l_m(t).
	\end{align*}	
	Summing over $t=1 \dots T$,
	\begin{align*}
	U\sum_{t=1}^T&\sum_{m,s=2}\Big(N^l_{m,s}(t-1)-\bar{N}^l_{m,s-1}(t)\Big)^+ \\
	\leq&\sum_{t=1}^{T}V\big[C_1^l(\bar{\eta}^l(t))-C_1^l(\bar{N}^l(t))\big] \\
	&+\sum_{t=1}^{T}\sum_mJ_m(t)\bar{N}^l_m(t)-\sum_{t=1}^{T}\sum_mJ_m(t)\bar{\eta}(t)^l_m(t)\\
	\leq& \sum_{t=1}^{T}V\big[C_1^l(\bar{N}^l(t-1))-C_1^l(\bar{N}^l(t))\big]\\
	&+\sum_{t=1}^{T-1}\sum_mJ_m(t)\bar{N}^l_m(t)-\sum_{t=2}^{T}\sum_mJ_m(t)\bar{N}^l_m(t-1) \\
	&+TMSA_{\max}N_{\max} \\
	\leq & TMSA_{\max}N_{\max} +\sum_{t=1}^{T-1}V\big[C_1^l(\bar{N}^l(t))-C_1^l(\bar{N}^l(t))\big] \\
	&+\sum_m\sum_{t=1}^{T-1}\Big[J_m(t)\bar{N}^l_m(t)-J_m(t)\bar{N}^l_m(t-1)\Big]\\
	\leq & TMSA_{\max}N_{\max} +\sum_m\sum_{t=1}^{T-1}\bar{N}^l_m(t)\Big[J_m(t)-J_m(t+1)\Big]\\
	\leq & TMSA_{\max}N_{\max} +\sum_m\sum_{t=1}^{T-1}\bar{N}^l_m(t)N_{\max}\\
	\leq & TMSA_{\max}N_{\max} +MTN^2_{\max}
ss	\end{align*}
	In the second inequality we used $\bar{\eta}^l(t) = \bar{N}^l(t-1)$ and the maximum workload till time slot $T$ is upper bounded by $TSA_{\max}$ for all job types. The second last inequality follows because $J_m(t)-J_m(t-1) \leq LN_{\max}$.	
	The average number of migrations is given by at each server is bounded by
	\begin{align*}
	\frac{1}{T}\sum_{t=1}^T\sum_{m,s=2}&\E\Big[\Big(N^l_{m,s}(t-1)-\bar{N}^l_{m,s-1}(t)\Big)^+\Big] \nonumber \\
	&\leq \dfrac{MN_{\max}(SA_{\max}+N_{\max}) }{U}.
	\end{align*}

\section{Pseudo code for Algorithm $1$}\label{Appenidix:Pseudocode}
Given $U,V,L, Q(t)$ and server's maximal VM configurations Algorithm \ref{Algo:Size} computes the job configuration using integer linear programming \eqref{DCN-eq:drift+penalty-control-action}. Once the job configurations for all the servers is computed the actual jobs that are placed on the servers is decided by the number of jobs waiting in the queues and previous job configurations on the servers. Accordingly queue lengths and previous job configurations are updated for the next slot. Note that, the job of type $(m,s)$ placed on the server $l$ at slot $t-1$ will be job type $(m,s-1)$ at slot $t$. So, in computing previous job configuration at slot $t$, using the job service configurations on the servers at slot $t-1$ we have to decrease the job sizes by one.
\begin{algorithm}
\caption{Pseudo code for Algorithm $1$ }
\label{Pseudo code}
\begin{algorithmic}[1]
    \State Given system parameters $V,U,L,Lambda,T$
    \State System Initialization
        \begin{align*}
            PrevConfig(l)& =0,\forall l\in [L]\\
            CurConfig(l) &=0, \forall l\in [L]\\
             Q(0) = 0,SR-Cost &= 0,JMC=0\\
        \end{align*}

        \While{$t<T$}  
        \State $Q(t) \leftarrow Q(t-1)+Arrivals(t-1)$  
        \State $PrevConfig(l) \leftarrow CurConfig(l), \forall l\in [L]$
        \State $CurConfig(l) \leftarrow from Algorithm 1,\forall l\in [L]$
        \State update $Q(t)$ according to the new job sizes
        \EndWhile
   	\end{algorithmic}
\end{algorithm}
\section{Proof of Theorem \ref{Theorem:UK-algo-1-optimality}} 
\label{Proof:Age-Algo-Optimality}
We combine the ideas in the proofs of \cite[Section V]{stheja-unknown} and Theorem \ref{Theorem:Size-Algo-Optimality} in this work. Let $\tilde{Q}_m(t)$ be the total type-$m$ jobs of all ages. Mathematically, $\tilde{Q}_m(t)=\sum_a\tilde{Q}_{m,a}, \forall m \in [M]$. For each $m$ let $S_m$ be a random variable representing size of a type-$m$ job. Let $W_m(a)$ be the expected remaining service time of a type-$m$ job given that it has been served for $a$ time slots. Mathematically, $W_m(a)=\E[S_m-a|S_m>a]$. We denote the expected workload backlog of type-$m$ jobs by $\tilde{J}_m(t)$. Thus 
\[\tilde{J}_m(t)=\sum_a\tilde{Q}_{m,a}W_m(a), \forall m \in [M].\]
The expected workload backlog evolves as \[\tilde{J}(t+1) = \tilde{J}_m(t)+\tilde{A}_m(t)-\sum_lD^l_m(t),\] where $\tilde{A}_m(t) = A_m(t)\bar{S}_m$ since each arrival of type-$m$ brings average workload $\bar{S}_m$. Moreover, $D^l_m(t)$ expected workload that has departed at time $t$, at server $l$ is computed as follows. Given state $(\tilde{J}(t), \tilde{A}(t), D(t))$, $\tilde{J}(t+1)$ is independent of previous states. So, $\{\tilde{J}(t)\}$ is a Markov chain. We investigate it's stability in the following.

\par Let $\tilde{p}_{m,a}=\P(S_m=a+1|S_m>a)$. A type-$m$ job that is scheduled for $a$ time slots has a workload backlog of $W_m(a)$. It departs in the next slot with probability $\tilde{p}_{m,a}$ and does not depart with probability $1-\tilde{p}_{m,a}$. So the departed workload for every a type-$m$ job $j$ that is being served for $a$ slots can be written as
\begin{align}\label{DCN-eq:dep_workload}
 \mathcal{D} ^j_m(t) = 
\begin{cases}
W_m(a) \text{  with prob } \hat{p}_{m,a}\\  
W_m(a)-W_m(a+1) \text{  with prob } 1-\hat{p}_{m,a}.
\end{cases}
\end{align} 
Moreover, $D^l_m(t) = \sum_{j=1}^{N^l_m(t)}\mathcal{D} ^j_m(t)$. From \eqref{DCN-eq:dep_workload} we see that $D^l_m(t)$ can be positive or negative. However, since  job sizes are bounded by $S$, $W_m(a) \leq S$ and $D^l_m(t \in [-N_{\max}S, N_{\max}S]$. Also, $\tilde{A}_m(t) \leq A_{\max}S$. Combining these we get 
\begin{align}\label{DCN-eq:ub-deltaJ}
\tilde{J}_m(t+1)-\tilde{J}_m(t) \leq A_{\max}S+N_{\max}S,
\end{align}
\begin{align}\label{DCN-eq:lb-deltaJ}
\tilde{J}_m(t+1)-\tilde{J}_m(t) \geq -N_{\max}S.
\end{align}
Since every job in queue has at least one more time slot of service remaining, $\tilde{J}_m(t) \geq \tilde{Q}_m(t)$. Further, $\tilde{J}_m(t) \leq S\tilde{Q}_m(t)$.
The following result is shown in \cite[Lemma 3]{stheja-unknown}.
	\par \textit{Fact:} If a type-$m$ has been scheduled for $a$ time slots, then the expected departure in the workload backlog is $\E[\mathcal{D}_m|l] = 1$. Therefore, we have $\E[\mathcal{D}_m] = 1$.
\par Let us define $g: [0,\infty) \to [0,\infty)$ as $g(x) = \log(1+x)$ and $G: [0,\infty) \to [0,\infty)$ as 
\[G(x) = \int_{0}^{x}g(y)dy.\]
We use the following Lyapunov function to establish stability of the expected workload backlog
$V(\tilde{J}(t)) = \sum_m G(\tilde{J}_m(t)).$
Define $\Delta(t)=V(\tilde{J}(t+1)-\tilde{J}(t))$. Then
\begin{align}\label{DCN-eq:delta-V1}
\Delta(t)& 
=\sum_m \Big(G(\tilde{J}_m(t+1))-G(\tilde{J}_m(t))\Big)  \nonumber \\
\leq &\sum_m \Big(\tilde{J}_m(t+1)-\tilde{J}_m(t))g(\tilde{J}_m(t+1)) \nonumber\\
= &\sum_m \Big(\tilde{J}_m(t+1)-\tilde{J}_m(t))\Big(g(\tilde{J}_m(t+1))-g(\tilde{J}_m(t))\Big)\nonumber \\
&+\sum_m\Big(\tilde{A}_m(t)-\sum_m D^l_m(t)\Big)g(\tilde{J}_m(t))
\end{align}
where the first inequality is due to convexity of $G(\cdot)$. Since $g'(.) \leq 1$, 
$|g(\tilde{J}_m(t+1)-g(\tilde{J}_m(t)| \leq |\tilde{J}_m(t+1)-\tilde{J}_m(t+1)|$.
 Thus the first term in \eqref{DCN-eq:delta-V1} can be bounded as
\begin{align*}
\sum_m& \Big(\tilde{J}_m(t+1)-\tilde{J}_m(t)\Big)\Big(g(\tilde{J}_m(t+1))-g(\tilde{J}_m(t))\Big) \\
&\leq\sum_m \Big|\tilde{J}_m(t+1)-\tilde{J}_m(t)\Big|\Big|g(\tilde{J}_m(t+1))-g(\tilde{J}_m(t))\Big| \\
&\leq \sum_m \Big|\tilde{J}_m(t+1)-\tilde{J}_m(t)\Big|^2 \leq K_1
\end{align*}
where $K_1=(MA_{\max}S+N_{\max}S)^2$ from \eqref{DCN-eq:ub-deltaJ}.
 Define $Y(t) = \big(\tilde{Q}(t),N(t-1),Z(t-1)\Big)$. Now taking conditional expectation, conditioned on $Y(t)$ of \eqref{DCN-eq:delta-V1} and using above we can write
\begin{align}\label{DCN-eq:cond-drift1}
\E\Big[&\Delta(t)|Y(t)\Big] \nonumber \\
\leq& K_1+\sum_mg(\tilde{J}_m(t))\lambda_m\bar{S}_m-\E\sum_{m,l}\Big[g(\tilde{J}_m(t))D^l_m(t)|Y(t)\Big].
\end{align}
Now we will bound the last term in \eqref{DCN-eq:cond-drift1}. Though Step $1$ in Algorithm~\ref{Algo:UK-algo-1} gives the VM configuration $\tilde{N}^l(t) l \in [L]$, there may be unused service when corresponding queue length is small. We will first bound this unused service. 
\par Define a fictitious departure process to account for the unused service as follows:
\begin{align}\label{DCN-eq:fict_dep_workload1}
 \tilde{\mathcal{D}}^j_m(t) = 
\begin{cases}
 { \mathcal{D} }^j_m(t)\text{ if }j\text{th job is served at time } t \\  
1 \text{  if } \text{otherwise,}
\end{cases}
\end{align} 
\begin{align}\label{DCN-eq:fict_dep_workload2}
\tilde{D}^l_m(t) = \sum_{j=1}^{\tilde{N}^l_m(t)}\tilde{\mathcal{D}}^j_m(t).
\end{align}
Now consider last term in \eqref{DCN-eq:cond-drift1} and use the similar approach as in obtaining \eqref{DCN-eq:N-bar} we get,
\begin{align}
\sum_{m,l}&g(\tilde{J}(t))D^l_m(t) \nonumber \\
&\geq \sum_{m,l}g(\tilde{J}_m(t))\tilde{D}^l_m(t)-MLN_{\max}g(LN_{\max}S) \nonumber
\end{align}
Using above inequality in \eqref{DCN-eq:cond-drift1}
\begin{align}\label{DCN-eq:cond-drift2}
\E\Big[&\Delta(t)|Y(t)\Big] \nonumber \\
\leq& K_2+\sum_mg(\tilde{J}_m(t))\lambda_m\bar{S}_m-\E\Big[\sum_{m,l}g(\tilde{J}_m(t))\tilde{D}^l_m(t)|Y(t)\Big] \nonumber \\
=& K_2+\sum_mg(\tilde{J}_m(t))\lambda_m\bar{S}_m-\sum_{m,l}g(\tilde{J}_m(t))\tilde{N}^l_m(t) 
\end{align}
where $K_2=K_1+MLN_{\max}g(LN_{\max}S)$. The last equality is due to \eqref{DCN-eq:fict_dep_workload1} and \eqref{DCN-eq:fict_dep_workload2}. Now adding server running and job migration costs on both sides of \eqref{DCN-eq:cond-drift2} gives	
\begin{align} \label{DCN-eq:cond-drift3}
\E\Bigg[&\Delta(t)+\sum_l\Big(C^l_1(V\tilde{N}^l(t))\nonumber \\&
+ U\tilde{C}^l_2(N^l(t-1),Z^l(t-1),\tilde{N}^l(t))\Big)|Y(t)\Bigg] \nonumber \\
\leq& K_2+\sum_mg(\tilde{J}_m(t))\lambda_m\bar{S}_m-\E\Bigg[\sum_{m,l}g(\tilde{J}_m(t))\tilde{N}^l_m(t)\nonumber \\
&-U\sum_{m,a=0}^{M,S-2}\Big(N^l_{m,a}(t-1)-Z^l_{m,a}(t-1)-\tilde{N}^l_{m,a+1}(t)\Big)^+\nonumber \\
&-\sum_lV\mathbbm{1}_{\{\tilde{N}^l(t)\ne 0\}}|Y(t)\Bigg].
\end{align}
We know 
\begin{align}\label{DCN-eq:ieq1}
g(\tilde{J}_m(t)) \leq &g(S\tilde{Q}_m(t))  \nonumber\\
=& \log(1+S\tilde{Q}_m(t)) \nonumber\\ 
\leq& \log(S(1+\tilde{Q}_m(t))) \nonumber\\
=& \log(S)+g(1+\tilde{Q}_m(t)),
\end{align}
and
\begin{align}\label{DCN-eq:ieq2}
g(\tilde{J}_m(t)) \geq g(\tilde{Q}_m(t)).
\end{align}
Using \eqref{DCN-eq:ieq1} and \eqref{DCN-eq:ieq2} in \eqref{DCN-eq:cond-drift3}
\begin{align} \label{DCN-eq:cond-drift4}
\E\Bigg[&\Delta(t)+\sum_l\Big(C^l_1(V\tilde{N}^l(t))\nonumber \\&
+ U\tilde{C}^l_2(N^l(t-1),Z^l(t-1),\tilde{N}^l(t))\Big)|Y(t)\Bigg] \nonumber \\
\leq& K_3+\sum_mg(\tilde{Q}_m(t))\lambda_m\bar{S}_m-\E\Bigg[\sum_{m,l}g(\tilde{Q}_m(t))\tilde{N}^l_m(t)\nonumber \\
&-U\sum_{m,a=0}^{M,S-2}(N^l_{m,a}(t-1)-Z^l_{m,a}(t-1)-\tilde{N}^l_{m,a+1}(t))^+\nonumber \\
&-\sum_lV\mathbbm{1}_{\{\tilde{N}^l(t)\ne 0\}}|Y(t)\Bigg].
\end{align}
where $K_3 = K_2+M\log(S)$. Now proceeding as in the proof of Theorem~\ref{Theorem:Size-Algo-Optimality} (See the steps following \eqref{DCN-eq:drift+cost_1}) we obtain,
\begin{align} \label{DCN-eq:UK_Result1}
\E\Bigg[&\Delta(t)+\sum_l\Big(C^l_1(V\tilde{N}^l(t))\nonumber \\
&+ U\tilde{C}^l_2(N^l(t-1),Z^l(t-1)\tilde{N}^l(t))\Big)|Y(t)\Bigg] \nonumber \\
\leq& K_3+LUN_{\max}-\e\sum_mg(\tilde{Q}_m(t))+V\bar{C}_{\text{opt}}(\lambda+\e_g) \nonumber \\
\leq & K_3+LUN_{\max}-\e\sum_mg\Bigg(\frac{\tilde{J}_m(t)}{S}\Bigg)+V\bar{C}_{\text{opt}}(\lambda+\e_g).
\end{align}
Since the costs are non-negative we can write the above inequality as,
\begin{align}\label{DCN-eq:UK_stability}
\E\Big[&\Delta(t)|Y(t)\Big] \nonumber \\
\leq & K_3+LUN_{\max}-\e\sum_mg\Bigg(\frac{\tilde{J}_m(t)}{S}\Bigg)+V\bar{C}_{\text{opt}}(\lambda+\e_g).
\end{align}	
Define $\mathcal{M} = K_3+LUN_{\max}+V\bar{C}_{\text{opt}}(\lambda+\e_g)$. Since the job size is bounded, we can find $\mathcal{B}=\{x:\e\sum_mg(x/S) < \mathcal{M}\}$ so that the expected Lyapunov drift is negative whenever $\tilde{J}\in \mathcal{B}^c$ and the system is stable\cite{stheja-unknown}.
Summing \eqref{DCN-eq:UK_Result1} over $t = 0,1,...,T-1$, on \eqref{DCN-eq:cond-drift3} we obtain the following bound on mean server running cost
\begin{align}\label{DCN-eq:UK_opt_serv_cost}
\lim_{T \to \infty} \sum_{t=1}^{T}\sum_l \E\big[C^l_1(\tilde{N}^l(t))\big] \leq& \dfrac{K_3+LUN_{\max}}{V} \nonumber \\
&+\bar{C}_{\text{opt}}(\lambda+\e_g).
\end{align}
Above inequalities \eqref{DCN-eq:UK_stability} and \eqref{DCN-eq:UK_opt_serv_cost} prove part $(a)$ and part $(b)$ of the theorem. 

 $\text{      }\text{      }$The proof of part~($c$) goes on the similar lines as Theorem~\ref{Theorem:Size-Algo-Optimality} part~($c$). Here also we consider at every $t$ another algorithm that does not cause job migrations at all. We compare the weight of the VM configuration provided by the two algorithms at every time slot. Further, we take the time average to bound job migrations. Consider the job migrations at server $l$ at time $t$. 
\begin{align}
U&\sum_{t=1}^{T}\sum_{m,a=0}^{M,S-2}\Big(N^l_{m,a}(t-1)-Z^l_{m,a}(t-1)-\tilde{N}^l_{m,a+1}(t)\Big)^+ \nonumber \\
\leq & B +\sum_m\sum_{t=1}^{T}\tilde{N}^l_m(t)\Big[g(\tilde{Q}_m(t))-g(\tilde{Q}_m(t+1))\Big] \nonumber \\
\leq & B +\sum_m\sum_{t=1}^{T-1}\tilde{N}^l_m(t)\Big[g(\tilde{Q}_m(t))-g(\tilde{Q}_m(t)-N_{\max})\Big] \nonumber \\
\leq & B +\sum_m\sum_{t=1}^{T-1}\tilde{N}^l_m(t)\log\Big[\frac{1+\tilde{Q}_m(t)}{1+\tilde{Q}_m(t)-N_{\max}}\Big] \nonumber \\
\leq & B
+\sum_{t=1}^{T-1}MN_{\max}\log\Big[1+N_{\max}\Big] \nonumber \\
\leq & B+TMN_{\max}\log\Big[1+N_{\max}\Big] \nonumber
\end{align}
where $B=MN_{\max}g\big(TA_{\max}\big)$. Take time average and expectation of above,
\begin{align}
\frac{1}{T}&\sum_{m,a=0}^{M,S-2}\E\Bigg[\Big(N^l_{m,a}(t-1)-Z^l_{m,a}(t-1)-\tilde{N}^l_{m,a+1}(t)\Big)^+\Bigg] \nonumber \\
\leq & \frac{1}{U}\Bigg(\dfrac{MN_{\max}g\big(TA_{\max}\big)}{T}+MN_{\max}\log\Big(1+N_{\max}\Big)\Bigg) \nonumber \\
\leq & \frac{1}{U}\Bigg(MN_{\max}A_{\max}+MN_{\max}\log\Big(1+N_{\max}\Big)\Bigg) \nonumber 
\end{align}

\end{document}